\documentclass{article}

\usepackage[preprint]{neurips_2018}

\usepackage[utf8]{inputenc} 
\usepackage[T1]{fontenc}    
\usepackage{hyperref}       
\usepackage{url}            
\usepackage{booktabs}       
\usepackage{graphicx} 
\usepackage{amsfonts}       
\usepackage{nicefrac}       
\usepackage{microtype}      
\usepackage{amsmath}
\usepackage{amstext}
\usepackage{algorithm}
\usepackage[noend]{algpseudocode}
\usepackage{amsthm}
\usepackage{natbib}
 \bibpunct[, ]{(}{)}{,}{a}{}{,}%
 \def\bibfont{\small}%
\usepackage{xcolor}

\newtheorem{theorem}{Theorem}
\newtheorem{defn}{Definition}

\newtheorem{example}{Example}
\newtheorem{question}{Question}
\newtheorem{answer}{Result}
\newtheorem{remark}{Remark}

\usepackage{mathtools}

\newcommand{\eps}{\epsilon}
\newcommand{\Vhat}{\ensuremath{\hat{V}}}

\newcommand{\xhat}{\ensuremath{\hat{x}}}
\newcommand{\Xhat}{\ensuremath{\hat{X}}}

\newcommand{\R}{\mathbb{R}}

\title{{\footnotesize Forthcoming in Operations Research (2021)}\\Computing Large Market Equilibria Using Abstractions}

\author{%
  Christian Kroer* \\
  IEOR Department\\
  Columbia University
  \And
  Alexander Peysakhovich* \\
  Facebook AI Research
  \AND
  Eric Sodomka \\
  Facebook Core Data Science \\
   \And
  Nicolas E. Stier-Moses \\
  Facebook Core Data Science
}

\begin{document}

\maketitle

\begin{abstract}
Computing market equilibria is an important practical problem for market design, for example in fair division of items. However, computing equilibria requires large amounts of information (typically the valuation of every buyer for every item) and computing power. We consider ameliorating these issues by applying a method used for solving complex games: constructing a coarsened abstraction of a given market, solving for the equilibrium in the abstraction, and lifting the prices and allocations back to the original market. We show how to bound important quantities such as regret, envy, Nash social welfare, Pareto optimality, and maximin share/proportionality when the abstracted prices and allocations are used in place of the real equilibrium. We then study two abstraction methods of interest for practitioners: (1) filling in unknown valuations using techniques from matrix completion, (2) reducing the problem size by aggregating groups of buyers/items into smaller numbers of representative buyers/items and solving for equilibrium in this coarsened market. We find that in real data allocations/prices that are relatively close to equilibria can be computed from even very coarse abstractions.
\end{abstract}

\maketitle

\section{Introduction}

The problem of computing the equilibrium prices and allocation of a market economy has large informational and computational requirements \citep{hayek1945use,scarf1967computation}. In this work we apply the idea of abstraction to market equilibrium computation. We construct a simplified model of a market, referred to as an abstraction, solve for the equilibrium in the abstraction, and then lift the solution back to the original market. We derive analytic bounds for the error in the computed equilibrium as a function of abstraction quality, describe two methods of abstraction that can be useful in practical problems, and evaluate our approach on real datasets.

Computing optimal allocations subject to constraints has been a problem of interest since the inception of modern economic theory. Early applications included trying to use linear programming to plan the entire Soviet economy \citep{kantorovich1960mathematical,kantorovich1975mathematics}. Modern market designers, with slightly less grand visions, also use centralized algorithms to construct allocations with various properties such as incentive compatibility, envy freeness, stability, efficiency, and notions of fairness \citep{roth2002economist,klemperer2018auctions}. 

We will focus on the canonical case of Fisher markets. In Fisher markets we have a set of divisible items to be allocated and a set of individuals who may receive the items (our theoretical results also hold for the indivisible setting). Individuals have a total budget of money and valuations for each item. Formally, a market equilibrium consists of a price for each item and an allocation such that (1) individuals cannot improve their utility by using their budgets to purchase a different set of items (given the prices) and (2) the total demand for each item (i.e., the sum of the solution of each individual's maximization problem of allocating their budget) is equal to the supply. The equilibrium in Fisher markets that maximizes Nash social welfare can be computed by solving the \emph{Eisenberg-Gale (EG) convex program} \citep{eisenberg1959consensus}.\footnote{Nash Social Welfare is of independent interest as a welfare criterion as it ``exhibits an elusive combination of fairness and efficiency properties'' \citep{caragiannis2016unreasonable}.}

An important reason for computation of market equilibria in practice is the problem of eliciting valuations and assigning items to individuals in the absence of money. One such mechanism is known as \emph{competitive equilibrium from equal incomes} or commonly simply CEEI~\citep{varian1974equity}. Individuals are endowed with a budget of pseudo currency, give their valuations for items, an equilibrium of the resulting market is computed, and individuals receive the allocation from this equilibrium. This mechanism is known to have good properties: it is envy free,\footnote{When goods are not divisible, an approximate version is envy-free-up-to-one-good \citep{budish2011combinatorial}.} achieves a Pareto-optimal outcome, and is incentive compatible when the market is large \citep{budish2011combinatorial,azevedo2018strategy}. A version of CEEI is used in practice in course allocations at business schools \citep{budish2016bringing}.

While the course allocation setting has complex preferences and constraints, we are primarily motivated by applications where the assumptions of linear additive utilities are appropriate and where scales are large.\footnote{We also assume divisible items, but this can instead be thought of as lotteries over items being permitted rather than items themselves actually being divisible.} One may consider using Fisher markets in recommender systems. Here buyers are content creators and items are impressions. An advantage of Fisher market based allocation, rather than standard click-through rate maximization, is that it creates a more equal distribution of impressions across creators which is desirable in many cases. One example is the recommendation of jobs on online market places. There, click-through rate maximization can lead to congestion issues where some jobs receive many applicants and others receive few. Such an extreme outcome is bad for both the job posters and the job seekers.

As another example, it is known that in online advertising, Fisher markets with quasi-linear utilities are directly related to the budget-management problem faced by large internet platforms that sell impressions using paced auctions~\citep{conitzer2018multiplicative,conitzer2018pacing}.\footnote{The standard Fisher model assumes that leftover budgets are useless outside the market. This assumption can be removed by making the utility function quasi-linear. In the quasi-linear Fisher market, equilibria can still be computed using a convex program \citep{cole2017convex}. A full discussion of division mechanisms with real money is beyond the scope of this paper, though we will show that our main allocation results hold in both standard and quasi-linear Fisher markets.}
These systems are typically online (e.g., using auction-based pacing). In this setting, the static model can either be used to warm start pacing multipliers off past data (as these multipliers can be derived from the EG program~\citep{conitzer2018pacing}), or to run a posted price-based system where prices are updated by computing static markets from past demand and supply.
In addition, computation of static Fisher market equilibria is of interest even in the online case for ex-post market analysis and counter-factual reasoning (``what would happen to prices if supply of $X$ were to go up by $Y \%$?'') in such markets. 
An interesting future direction is to consider how to solve the Fisher-market convex program in an online fashion. In an online setting some of our proposed abstraction methods could be combined with an online algorithm directly, for example our low-rank market model.

Scaling market equilibrium computation to larger markets, such as internet-scale ones, requires solving two major challenges. First, one needs access to large amounts of information (e.g., every person's valuation for every item or item combination). Second, one needs access to enough computing power to solve the related convex program which is solvable in polynomial time \citep{nisan2007algorithmic} but still unwieldy for very large markets (e.g., thousands or millions of buyers). 

In this paper we propose a solution methodology for finding the equilibria of large markets using \textit{market abstractions}. We use the available details of a market instance to create a compact abstraction (a simplified model of the market), solve for the equilibrium in the abstraction, and then lift the solution to the original market. We provide a general set of results that allow us to bound various metrics of relevance that capture the distance to equilibria for the allocation and prices derived from the abstraction. This is the best that one can hope for large instances since computing the true equilibrium is computationally hard.
Some metrics for which our bounds are valid include regret, envy, Pareto optimality, and Nash social welfare.
We show that these results hold both for standard and quasi-linear Fisher markets.

We then turn to practice and focus on two abstraction methods of particular interest. The first uses techniques from matrix completion to infer valuations for person-item pairs that we may not have access to. We refer to this case as \textit{low-rank markets}. The second consists in reducing the size of the market (and thus the computational burden) by replacing groups of buyers and items with \textit{representative buyers} and \textit{representative items}, respectively. We solve for equilibria in the space of representative buyers/items and then split the allocation of each representative among the buyers/items it represents. We refer to this case as \textit{representative market equilibrium}. We show that this abstraction reduces the computational complexity and can be efficiently parallelized. We provide two lifting procedures to go from the abstract solution to the original market---a proportional and a recursive version---and discuss the tradeoffs between them.

We apply these abstraction methods, which can be used together, to real datasets including a novel one which we have collected, and evaluate the quality of solutions for various abstraction levels. 
We find that the equilibria found even in coarse abstractions have low errors in envy, regret, maximin share, distance to Pareto optimality, and Nash social welfare.


\section{Related Work}

The use of equilibrium assumptions to make estimates of deep `structural' parameters is of interest in both applied micro and macroeconomics \citep{berry1995automobile,ljungqvist2018recursive}. Often, the lack of individual-level data, or computing power, requires the use of representative agents (i.e., a single consumer that represents all individuals in the economy). To get around this issue, analysts typically make (strong) assumptions which imply that the equilibrium prices, aggregate decisions, and some structural parameters computed using a representative-agent stand-in are equivalent to the ones that would be derived from a model which includes all individual agents. In our work we are specifically interested in the situation where these assumptions are not true and the lifted answer from the abstraction does not yield the true equilibrium prices/allocation. There is increasing interest in using heterogeneous agent models in applied economic modeling \citep{hommes2006heterogeneous} and expanding our results to these cases is an interesting direction for future work.

Abstraction is an idea often used in the context of games. In two-player, zero-sum games (e.g., poker) a popular goal is to compute a Nash equilibrium strategy during training time and use it in actual play against opponents. In large games, however, it is impractical to solve for the Nash equilibrium of the original game and practitioners solve for the Nash equilibrium of an abstraction and then lift it to be a strategy for the original game. Abstractions often use heuristics and are hand tuned \citep{gilpin2007potential,waugh2009practical,ganzfried2014potential,brown2015hierarchical,brown2018superhuman} but they have more recently begun to be constructed automatically using function approximation (e.g., deep learning) \citep{moravvcik2017deepstack,brown2018deep}. In games, unlike in markets, the relationship between the quality of abstraction and the quality of the lifted strategy in the original game has been heavily studied \citep{lanctot2012no,kroer2014extensive,kroer2016imperfect,kroer2018unified}.

There is a literature on abstraction in non-market-based allocation problems \citep{walsh2010automated,peng2016scalable,lu2015value}. The scalability problem faced in the allocation setting is similar to ours,  but because the underlying optimization problem is very different (i.e., the maximization of allocative efficiency rather than a market equilibrium) the abstraction methods and results in these papers are quite different in character from ours.

There is a large literature on computing market equilibria in Fisher markets using convex programming \citep{eisenberg1959consensus,shmyrev2009algorithm} or gradient-based methods \citep{birnbaum2011distributed,nesterov2018computation}, typically via the \emph{Eisenberg-Gale convex program} (EG). There is also work extending the EG program to new settings such as quasi-linear utilities, indivisible items, or set-valued utility functions \citep{chen2007note,caragiannis2016unreasonable,cole2017convex,cole2018approximating,murray2020robust}. Our paper complements this existing work as our results are agnostic to the algorithm used for equilibrium computation; any of these algorithms can be employed for computing a market equilibrium in conjunction with our market abstraction model.

The literature on utilitarian welfare maximization in the context of market equilibrium is related to the quasilinear variant of the Fisher-market setting that we consider. In particular, that literature is motivated by applications to auctions, where market equilibrium with quasilinear utilities is a natural concept. 
Typically, this literature focuses on indivisible items, but does not allow budget constraints.
For example, \citet{baldwin2019understanding} study such a setting motivated by the auctions used for auctioning off loans during financial crises. 
\citet{leme2020computing} study fast oracle-algorithms for computing market equilibria in this setting when only an aggregate demand oracle is available.
\citet{bichler2019computing} study an exchange setting with utilitarian welfare maximization, where budgets are allowed. However, the problem of computing a market equilibrium in that setting turns out to be significantly harder than in our setting ($\Sigma^p_2$-hard in particular).
See \citet{bichler2020walrasian} for a survey of the general literature on market equilibrium with utilitarian welfare.
In this paper, we focus on the case of Fisher markets with budgets, where the EG convex program and its quasilinear extension naturally leads to market equilibrium. 
An interesting direction for future work would be to apply abstraction techniques to more general market equilibrium problems under utilitarian welfare maximization.

Another reason for computing market equilibria is to perform counterfactual reasoning in complex systems such as online advertising \citep{bottou2013counterfactual,peysakhovich2019robust}. For example, it is known that market equilibrium is also the equilibrium that results from budget smoothing in first and second-price auction markets such as those for online advertising \citep{borgs2007dynamics,conitzer2018multiplicative,conitzer2018pacing}. Thus, computing EG equilibria in large markets gives a way to answer counterfactual questions like: ``what would happen to prices if supply of some items changed by $10\%$?'' using only supply and demand information without having to simulate the actual rules of the market.

A related setting to ours is the one with linear additive utilities but indivisible items. In such settings, market equilibrium is not guaranteed to exist, but a closely related solution concept is to compute the discrete allocation that maximizes the EG objective (i.e., simply discretize the EG program). This computation is used in allocation mechanisms on Spliddit, an online application which helps groups of individuals allocate items fairly~\citep{goldman2015spliddit}, motivated by problems such as fairly splitting jewelry, artworks, or even an entire estate. The site also allows some of the goods to be divisible (e.g., cash or financial assets). Our theoretical results also apply to this setting, though our algorithmic results for generating smaller markets via clustering would require changes in order to yield a reasonable solution approach.


\section{Market Equilibrium in Fisher Markets}

In this paper we focus on the canonical \emph{Fisher market} setting with linear utilities. We have a set of $n$ buyers and a set of $m$ items (or goods; we will use the two terms interchangeably). Each item $j$ has a supply of $s_j$ units that are divisible. We denote by $X\in \R_+^{n\times m}$ the allocations of items to buyers where $x_{ij}$ refers to the share of item $j$ allocated to buyer $i$. We say an allocation is a full allocation if the sum of each column $j$ of $X$ is $s_j$. We use the convention that a vector $x_i$ corresponding to buyer $i$ and denoted in lower case is the $i$-th row of $X$, and use the shorthand $[n] = \{1,\ldots, n\}$ to denote the index set of buyers or items.

Each buyer $i$ is endowed with a budget $B_i$. We denote by $V \in \R_+^{n\times m}$ the matrix of values with $v_{ij}$ being the value of buyer $i$ for item $j$ and $v_i$ the vector of values for buyer $i$.  For a given allocation $X$ we assume that total values are additive and linear, i.e., 
$$u_i(x_i) = v_i \cdot x_i = \sum_{j} v_{ij} x_{ij}.$$
We assume throughout that each item $j$ has at least one buyer $i$ such that $v_{ij}>0$, and each buyer $i$ has at least one item~$j$ such that $v_{ij}>0$. Note that if this does not hold then we may preprocess away such degenerate buyers and items.

\begin{defn}
Given (linear, anonymous) prices $p \in \R_+^m$ for the items, the \emph{demand} set for buyer $i$ is 
  $$d_i (p) = \lbrace \arg \max_{x_i \in \R_+^m} v_i \cdot x_i : p \cdot x_i \leq B_i, x_i \leq s\rbrace.$$ 
\end{defn}
Note that in our definition of demand computation each buyer is \emph{supply-aware}: they require $x_i\leq s$ in their computation. This is not crucial for any of our results and we could also remove that constraint from the definition of $d_i(p)$. Including it is motivated by the fact that supply-aware buyers allow market equilibria to capture budget-smoothing in second-price auction markets~\citep{conitzer2018multiplicative}. 

\begin{defn}
A set of prices and an allocation $(p^*, X^*)$ is a \emph{market equilibrium} if $x^*_i \in d_i(p^*)$ for all $i\in[n]$ and $X^*$ is a full allocation.
\end{defn}

A market equilibrium is, in general, hard to compute \citep{chen2009settling,othman2016complexity}, even in the piecewise-linear Fisher case~\citep{chen2009spending}. Furthermore markets may have multiple equilibria in the case of supply-aware buyers. However, in linear Fisher markets a particular market equilibrium is computable using the Eisenberg-Gale (EG) convex program \citep{eisenberg1959consensus}:

\begin{equation*}
\begin{aligned}
& \underset{X\in \R_+^{n\times m}}{\text{max}}
& & \sum_{i} B_i \text{log} (v_i \cdot x_i) \\
& \text{subject to}
& & \sum_{i} x_{ij} \leq s_{j} \text{ } \forall j.
\end{aligned}
\end{equation*}

This convex program has no incentive constraints and no prices, so it may be surprising that it is somehow related to a market equilibrium. However, it can be shown that the optimal solution $X^*$ for this convex program is a market equilibrium allocation. The corresponding prices are the dual variables associated with each of the supply constraints.

An attractive feature of the EG equilibrium is that it tells us exactly what the equilibrium outcome actually maximizes---the geometric mean of individual utilities weighted by budgets---i.e., the Nash social welfare when budgets are $1$~\citep{varian1974equity,caragiannis2016unreasonable}. For this reason the EG solution is also referred to as the \emph{max Nash welfare} (MNW) allocation.

The equilibrium found by the EG procedure has a useful structure: for every item that an individual is receiving in their equilibrium allocation they receive equal `bang per buck.' Formally, in the EG equilibrium $(p^{*}, X^{*})$ $$\text{if } x^{*}_{ij}, x^{*}_{ij'} > 0 \text{ then } \dfrac{v_{ij}}{p^{*}_{j}} = \dfrac{v_{ij'}}{p^{*}_{j'}}.$$
This also corresponds to the unique equilibrium in the case where buyers are not supply-aware.
From the convex program it can be shown that the set of EG prices and per-buyer utilities are unique. On the other hand, there may be multiple associated equilibrium allocations, for example, if buyers view two items as perfectly interchangeable. 

\subsection{Quasi-Linear Fisher Markets}

So far we have considered Fisher markets where money has no intrinsic value. These are useful for fair division applications such as the computation of CEEI where the `money' allocated to each individual is indeed useless outside the mechanism. Another type of Fisher market is a \textit{quasi-linear Fisher market} where leftover budget is useful outside of the market. For example, if we are considering the case of online ad markets, any unspent budget is refunded to the buyers and can be used in outside purchases. Formally, this changes the utility function to depend on prices as follows:
$$u_i(x_i, p) = v_i \cdot x_i + (B_i - x_i \cdot p).$$
The definition and existence of equilibrium prices and allocations can be derived just as above. In addition, it is known that such quasi-linear Fisher market equilibria can also be found via a slightly modified version of the EG convex program \citep{cole2017convex}.

To keep notation simple, we will focus on the standard Fisher market case for the theory section of this paper. However, in the Appendix we show that all of the results in terms of bounds on the error introduced by abstraction readily extend to the case of quasi-linear Fisher markets. In addition, we perform an experiment in the quasi-linear setting and show that, just as in the standard Fisher setting, abstractions give good results in practice.


\section{Abstractions of Markets}
We are interested in settings where we do not have full access to the valuation matrix $V$.
Let us call the original market $M = (V,B,s)$.
Instead of $M$, we consider some smaller market $\tilde M = (\tilde V, \tilde B, \tilde s)$ which is an abstraction of $M$. The number of buyers and items in $\tilde M$ may be smaller than that of $M$. 
A market equilibrium $(\tilde p, \tilde X)$ is then computed for the abstraction $\tilde M$.
Finally, we specify a \emph{lifting procedure} that generates a set of prices $\hat p$ and an allocation $\hat X$ for $M$, based on the abstraction solution $(\tilde p, \tilde X)$. This lifting procedure also produces a valuation matrix $\hat V$ of the same dimensions as $V$, such that $(\hat p, \hat X)$ is an equilibrium under the market $\hat M = (\hat V, B, s)$. The purpose of the matrix $\hat V$, which encodes the abstraction and lifting procedures, is that we will use it to derive bounds on the solution quality of $(\hat p, \hat X)$ in terms of $V$.

To be concrete, the abstract market $\tilde M$ could arise from a clustering of the buyers and items: every buyer is assigned to a representative buyer, and every item is assigned to a representative item. In that case, the approximate valuation matrix $\hat V$ would be such that each value $\hat v_{ij}$ is equal to the value $\tilde v_{\tilde i,\tilde j}$ of the corresponding representative buyer $\tilde i$ and representative item $\tilde j$.

Given this setup, our main question is:
\begin{question}[Main Question]
If $(\hat{p}, \hat{X})$ is an equilibrium with respect to $\hat{V}$ what can be said about it relative to the true valuation matrix $V$?
\end{question}

\begin{defn}
We define the abstraction error as $\Delta V = V - \hat{V}$. We use $\Delta v_i$ to denote $v_i - \hat v_i$.
\end{defn}

Throughout this section and the section on approximation results we assume that $s_j=1$ for all $j$. This is only for ease of notation, and is without loss of generality as valuations and prices can be rescaled to fit the assumption. The next section will be focused on showing a rather technical and general theorem which implies the following results:

\begin{answer}[Main Results, Informal Statement] 
For commonly used evaluation metrics of candidate solutions such as regret, Nash social welfare, envy, Pareto optimality, and maximin share, the error from using $\hat{V}$ instead of $V$ can be bounded using various matrix norms of $\Delta V$.
\end{answer}

Individual bounds will largely be in terms of $\|\Delta v_i\|_1$, the $\ell_1$-norm of the change in values for buyer $i$. Bounds over all agents will mostly use the $\ell_1,\ell_{\infty}$-norm for matrices, but where the $\ell_1$ part is applied to rows rather than columns: $$\|\Delta  V\|_{1,\infty} = \max_{i \in [n]} \|\Delta v_i\|_1.$$
Note that the standard definition of this norm is to apply it to columns, not rows. 

Again, we note that while this paper primarily focuses on the divisible Fisher market setting, our results in this section all carry over to the indivisible setting as well as to quasi-linear Fisher markets.

\subsection{LQFS Properties Under Market Abstractions}

We now define a general class of properties of allocations and prices, referred to as \emph{linear and quantified over a fixed set} (LQFS), that will be approximately preserved under abstractions. Our definition is a generalization of the \citet{balkanski2014simultaneous} notion of linear properties in fair division mechanisms (see Appendix~\ref{sec:lin props} for why the \citet{balkanski2014simultaneous} definition does not give us what we want). 

  We give an abstract definition of a property that we use to prove a general approximation result. We then instantiate it for several types of guarantees such as no envy, which compares the current allocation across agents, and global properties such as the potential for Pareto improvement.
  The intuition is relatively simply: we wish to capture properties such that if the property holds in the abstracted market with valuation matrix $\hat V$, then the error in the property when measured under $V$ is linearly related to the error $\Delta V$. This will more-or-less follow from the linearity of utilities.
  However, a second crucial concern is the set of comparisons that need to be checked, in order for the property to hold. Intuitively, this set of comparisons should remain invariant under changes to the valuation matrix $V$. 

\begin{defn}
  A property $k$ of an arbitrary tuple of valuations, prices, and allocations $ (V, p, X)$ is per-buyer LQFS (linear and quantified over a fixed set) if it can be represented as follows: We are given some set of allocations $\mathcal X^k$ that we compare to, a quantifier $Q \in \{\exists, \forall\}$, weights $\lambda_1^k, \lambda_2^k \geq 0$, and sets of buyers $I_i$ for each $i\in [n]$. We then ask that for all $X' \in \mathcal X^k$, $Qi \in [n]$ (i.e., either for all $i$, or there exists $i$), we have
\begin{align}
  \label{eq:min_sublinear_prop}
   u_i(x_{i}) \geq \lambda_1^k \max_{i'\in[n]} u_i(x_{i'}) +  \lambda_2^k \min_{i' \in I_i} u_i(x_{i'}'),
\end{align}
where $I_i$ is the set of buyers we wish to compare $i$'s allocation to.

A property $k$ of $ (V, p, X)$ is global LQFS if, given $w_i^k \geq 0$ for each $i\in [n]$, it can be represented as follows
\begin{align}
  \label{eq:lqfs global}
   \sum_{i\in [n]} w_i^k u_i(x_{i}) \geq \sum_{i\in [n]} w_i^k u_i(x_i') \ \forall X' \in \mathcal{X}^k.
\end{align}

In both cases, $\mathcal{X}^k$ is a subset of all supply-feasible allocations, which is allowed to depend on $p$ and $X$, but not $V$. 
\end{defn}

Equation~\eqref{eq:min_sublinear_prop} compares the utility of buyer $i$ under the current allocation and prices to various other allocations:
the first term on the right-hand side with weight $\lambda_1^k$ is used to measure the extent to which buyer $i$ likes the bundles assigned to other buyers under $p,X$; the second term with weight $\lambda_2^k$ measures how much buyer $i$ improves under the alternative allocation $X'$, if forced to take the worst allocation in the set $I_i$. Practically speaking, $I_i$ is typically either equal to $[n]$ or $\{i\}$.

The goal will be to show that each LQFS property on $(\hat V,\hat p, \hat X)$ is an approximate LQFS property on $(V, \hat p, \hat X)$.
Crucially, the sets $\mathcal{X}^k_i$ in per-buyer LQFS properties and $\mathcal X^k$ in global LQFS properties will remain the same both in the abstraction and in the original market since they do not depend on $V$; for most properties we examine, this set is either a singleton, the set of all supply-feasible allocations, or the set of budget-feasible allocations under $p$.

We now show our main technical result:

\begin{theorem}\label{lqfsworks}
  LQFS properties are approximately preserved under abstraction in the following sense: 
  
  Consider some particular per-buyer LQFS property $k$ which is satisfied under $(\hat{V},\hat{p},\hat{X})$.
  Then the property is approximately satisfied for $(V,\hat{p},\hat{X})$. In particular, for all $X' \in \mathcal X^k$, we have that  $Qi \in [n]$:
  \begin{align*}
   u_i(\hat x_{i}) \geq \lambda_1^k \max_{i'\in [n]} u_i(\hat x_{i'}) +  \lambda_2^k \min_{i' \in I_i} u_i(x_{i'}') - \max(1, \lambda_1^k +\lambda_2^k) \|\Delta V\|_{1,\infty}.
  \end{align*}

  Consider some global LQFS property $k$, quantified over a set of allocations $\mathcal{X}^k$ which is satisfied under $(\hat{V},\hat{X},\hat{p})$.
  Then the property is approximately satisfied for $(V,\hat{p},\hat{X})$. In particular,
  \begin{align*}
   \sum_{i\in [n]} w_i^k u_i(\hat x_{i}) \geq \sum_{i\in [n]} w_1^k u_i(x_{i}') - \sum_{i \in [n]}w_i^k \|\Delta v_i\|_{1}, \ \forall X' \in \mathcal{X}^k.
  \end{align*}

  \label{thm:lqfs}
\end{theorem}

The proof is given in the Appendix.

\subsection{Which Properties Are LQFS?}
We now show that a large number of market equilibrium properties are in fact LQFS properties and thus they are approximately preserved under abstraction.

One notion of the quality of a pair $(p,X)$ of prices and allocation is the induced regret. That is, how close is $\hat{x}_i$ to actually being a demand vector given the prices $\hat{p}$? Here we measure regret with respect to the supply constraints when the buyer can buy all items for themselves. Formally the regret of a buyer under a solution $(\hat{p}, \Xhat)$ is
\[\text{Regret}_i (\hat{p}, \hat{X}) = u_i(d_i(\hat p))  - u_i(\xhat_i).\]
By definition, in equilibrium, regret is $0$.

Another important metric of allocation quality is envy \citep{varian1974equity,budish2011combinatorial,caragiannis2016unreasonable}. The envy for buyer $i$ is the amount by which they prefer any other buyers' allocation over their own. Importantly, when buyers have the same budget, equilibrium allocations are envy free. This is a major reason for the popularity of CEEI as an allocation mechanism. Formally, the envy is
\[\textrm{Envy}_i(\hat X) = \max_{i'\in [n]} u_i( \hat x_{i'}) - u_i(\hat x_i).\]

Finally, we look at the maximin share (MMS) guarantee. MMS is the value that buyer $i$ would obtain if they were allowed to split the allocation $X$ into $n$ bundles, but were then assigned their least-valued bundle $\min_{i'}u_i(x_{i'})$. MMS was introduced by \citet{budish2011combinatorial}, and it generalizes the \emph{proportional-share} notion of divisible assignment to the indivisible setting. In proportional-share assignment, each buyer $i$ is required to receive an allocation that they like at least as much as the one where they receive a fraction $\frac{1}{n}$ of each item. Since MMS generalizes fair share our result below applies to both fairness notions.

Formally, the MMS guarantee of buyer $i$ is
\[
  \textrm{MMS}_i = \max_{x \in X} \min_{i' \in [n]} u_i(x_{i'}).
\]

We say that the \emph{MMS gap} for buyer $i$ in allocation $x$ is
\[\textrm{MMS gap}_i(\Xhat)=\max (0, \textrm{MMS}_i - u_i(\xhat_i)).\]

Often in market design the goal is to maximize social welfare, i.e., the sum of buyer utilities. However, in general, a market equilibrium may not lead to a social-welfare-maximizing solution. This is easy to see by considering the solution to the problem of maximizing geometric-mean utility: it is a market equilibrium, but it will never set any buyer utility to zero, even when that is  the only way to maximize social welfare. 

Instead, the usual efficiency criterion for market equilibria is that they are Pareto optimal (see, e.g., \citet{varian1974equity,budish2011combinatorial,caragiannis2016unreasonable}). We focus specifically on ability to strongly Pareto improve an allocation $\hat{X}$: for any Pareto-improving allocation $X'$, we say that its \emph{strong Pareto improvement} is the utility gain of the buyer who gained the least. For a market equilibrium, we are guaranteed that there does not exist any Pareto-improving allocation, and thus there is also no strongly-improving allocation.

While market equilibrium has no guarantees regarding social welfare, it can nonetheless be related to \emph{weighted} social welfare maximization. Negishi's theorem~\citep{negishi1960welfare} says that given a market equilibrium, there exists a set of utility weights $\{\beta_i\}$ specifying the utility weight of each buyer such that the market-equilibrium allocation is the allocation that maximizes weighted social welfare (in Fisher markets these weights are simply the inverse bang-per-bucks). 

Our next result shows that each of these commonly studied properties falls into the LQFS class. 

\begin{theorem}
Regret, envy, MMS, proportional share, and strong Pareto improvement are per-buyer LQFS properties. Weighted social welfare maximization is a global LQFS property.
\end{theorem}

\begin{proof}
  Throughout the proof, we let $(p,X)$ be the allocation and prices satisfying a given property for $V$. 

  For regret, we instantiate LQFS with the comparison set $\mathcal X^k$ equal to all supply-feasible allocations that are budget-feasible for every buyer under $p$,  quantifier $Q=\forall$, weights $\lambda^k_{1}=0, \lambda^k_{2}=1$, and $I_i = \{i\}$.
  This yields that for all $i$ and all supply-feasible allocations $X'$ such that $x_{i'}'\cdot p \leq B_{i'}$, we have the inequality (note that this includes supply and budget-feasible allocations which allocate only to $i$)
  \[
    u_i(x_i) \geq u_i(x_{i}').
  \]

  For envy, we instantiate LQFS with comparison set $\mathcal{X}^k = \{X\}$,  quantifier $Q=\forall$, weights $\lambda^k_{1}=1, \lambda^k_{2}=0$, and comparison set $I_i=\emptyset$. 
  This yields that for all $i$, we have the inequality
  \[
    u_i(x_i) \geq \max_{i'\in [n]} u_i(x_{i'}).
  \]

  For MMS and proportionality, we instantiate LQFS with the comparison set $\mathcal X^k$ equal to the set of all  supply-feasible allocations,  quantifier $Q=\forall$, weights $\lambda^k_{1}=0, \lambda^k_{2}=1$, and $I_i = [n]$. This instantiation gives MMS when goods are indivisible, and proportionality when goods are divisible. In particular, we get that for all $i$ and every supply-feasible allocation,
  \[
    u_i(x_i) \geq \min_{i'\in [n]} u_i(x_{i'}').
  \]
  In the divisible case the proportional allocation $x_i = \frac{1}{n}$ is included in this set, and every other supply-feasible allocation has some $i'$ with lower $u_i(x_{i'}')$ than in the proportional allocation. In the MMS case, this is exactly the same as the MMS definition. 

  For strong Pareto improvement, we instantiate LQFS where the comparison set $\mathcal X^k$ is the set of all supply-feasible allocations,  quantifier $Q=\exists$, weights $\lambda^k_{1}=0, \lambda^k_{2}=1$, and $I_i = \{i\}$.
  This guarantees that for every supply-feasible allocation $X'$, there exists some $i$ such that
  \[
    u_i(x_i) \geq u_i(x_i'),
  \]
  which is equivalent to saying that there is no strong Pareto improvement.

  For weighted social welfare maximization with weights $\beta \in \mathbb R_+^n$, we instantiate global LQFS with the comparison set $\mathcal X^k$ equal to the set of all supply-feasible allocations and weights $w^k_i = \beta_i$ equal to the weights we use in our weighted social welfare objective.
  This gives that for all $X'$ we have
  \[
    \sum_i w^k_i u_i(x_i)  \geq \sum_i w^k_i u_i(x_i'),
  \]
  which is exactly the guarantee needed for maximizing weighted social welfare.
\end{proof}

Using the fact that each of our measures of interest is an LQFS property, the main result follows immediately from Theorem \ref{lqfsworks}.
First we cover the per-buyer LQFS properties.

\begin{theorem}
  Let $(\hat{p},\hat{X})$ be equilibrium allocations and prices computed under $\hat{V}$. We can give the following bounds for each of the properties under the true valuation matrix $V$:
\begin{enumerate}
\item The regret, envy, and MMS gap for each buyer is at most $\| \Delta V \|_{1, \infty}$ 
\item Any strong Pareto improvement from $\hat{X}$ is bounded by $\| \Delta V \|_{1, \infty}.$
\end{enumerate}
\end{theorem}

Since weighted social-welfare maximization is a global LQFS property, we get the following.
\begin{theorem}[Negishi Welfare is Preserved Under Abstraction]
Let $\{\beta_i\}$ the Negishi social-welfare weights for the equilibrium under $\hat{V}$. Then $\Xhat$ solves the problem of maximizing social welfare for $V$ under weights $\{\beta_i\}$ up to an additive error of at most 
$\sum_{i \in [n]}\beta_i \|\Delta v_i\|_{1} \leq \|\beta\|_1\|\Delta V\|_{1,\infty}$.
\end{theorem}

\subsection{Properties that are not LQFS}

Another global optimality condition that has gained recent interest from the perspective of fairness is Nash social welfare \citep{caragiannis2016unreasonable}. The Nash social welfare (NSW) is the product of buyer utilities, formally
\[
  \text{NSW}(X) = \prod_{i\in [n]} u_i(x_i).
\]
NSW is appealing as compared to regular social welfare, since it ensures that we must give all buyers a reasonable amount of utility, or they will cause the whole objective to be small.
In fact, the EG program computes the solution that maximizes NSW, so market equilibria themselves are exactly the set set of NSW maximizers.

While NSW is not an LQFS property, we can show that it is approximately preserved under abstraction:
\begin{theorem}[NSW is Preserved Under Abstraction]\label{thm:nsw}
  \label{thm:nsw}
  The NSW of the optimal EG solution $\Xhat$ under $\Vhat$ upper bounds the NSW of the optimal solution $X^*$ to the original EG problem as follows:
  \[
    \text{NSW}(X^*) \leq \prod_{i \in [n]}\bigg(1 +  \frac{\|\Delta v_i\|_1}{\hat u_i(x_i^*)}\bigg) \text{NSW}(\xhat).
  \]
  This assumes $\hat u_i(x^*_i)>0$ for all $i$, i.e., the difference between $V$ and $\Vhat$ should be such that the original optimal solution still has nonzero value under $\Vhat$.
\end{theorem}

In contrast to the various properties above, the improvement to social welfare in any Pareto-improving allocation is not bounded.  Consider the following counterexample:
\[
  V =
  \begin{bmatrix}
    1 & \epsilon & \epsilon\\
    0 & 1 & \epsilon
  \end{bmatrix}, \quad
  \hat V =
  \begin{bmatrix}
    1 & \epsilon & 0\\
    0 & 1 & \epsilon
  \end{bmatrix}, \quad
  B_1 = B_2 = 1, \quad
  s_1=s_2=s_3=1.
\]
In the abstracted valuations $\hat V$ the unique market equilibrium has prices $\hat p = (0,2,0)$ and allocates item 1 to buyer 1, item 3 to buyer 2, and splits item 2 in half among the buyers. But now, in spite of $\|\Delta V\|_{1,\infty}=\epsilon$, we get that we can improve social welfare by $\frac{1}{2}-\epsilon$ under a Pareto-improving allocation by splitting item 3 instead of item 2.

The reason can be seen through the lens of LQFS properties. Maximizing social welfare over Pareto-improving allocations can be thought of as  solving the linear program $\max_{x'} \{\sum_i u_i(x_i') : u_i(x_i') \geq u_i(x_i)\ \forall i, \sum_i x_{ij}' \leq 1\ \forall j\}$, where $x_i$ is the allocation of buyer $i$ in the allocation we are trying to Pareto improve. If $\hat X$ is Pareto optimal under $\hat V$ and the utility functions in the constraints are defined with respect to $\hat V$, then this can be thought of as quantifying over the set of allocations satisfying the constraints in the LP, and indeed we are guaranteed that the objective at $\hat X$ is higher than for any other allocation in this quantification (under $\hat V$).
But when we use $V$ to define the utility functions in the constraints, we change the set of allocations we compare to (in other words, the quantification is not fixed when moving from $\hat V$ to $V$).

\subsection{Tightness of Bounds}
We now show that our bounds are tight. Consider the following market:

\begin{example}
There are $n\geq 2$ buyers and $nm$ items with supply equal to one, with $m \geq 1$. Each buyer $i$ has value $1$ for every item, except items $(i-1)m+1,\ldots,im$, for which they have value $1+\epsilon$. Each buyer has budget $m$. Now consider an abstraction $\hat V$ where every buyer has value $1$ for every item, this abstraction has $\|\Delta V\|_{1,\infty}= m\epsilon$, since every buyer has $m$ items for which their utility is lowered by $\epsilon$. Under $\hat V$, every market equilibrium has all prices equal to $1$, and any allocation that allocates exactly $m$ units of items to every buyer is a corresponding market equilibrium allocation. Consider the allocation where buyer $i$ gets all of items $im+1,\ldots (i+1)m$ (with buyer $n$ getting the first $m$ items); this is a market equilibrium allocation under $\hat V$ with all prices set to $1$. But now $i$ envies buyer $i-1$ by $m\epsilon$, and has regret $m\epsilon$. Furthermore, by rotating the allocation we can improve every buyer's utility by $m\epsilon$, and improve Negishi-weighted social welfare by $\|\beta\|m\epsilon$ (note that the Negishi weights must all be the same in this market). This example works for any $n\geq 2,m\geq 1$, and thus shows asymptotic tightness. 
\end{example}


\section{Abstractions in Practice}

The two major obstacles to computing equilibria in practice are information requirements (needing to know every element of $V$) and computation requirements. We now describe two techniques---matrix completion and the use of representative buyers/items---that can be used to reduce these burdens. We discuss how they fit into the framework of abstractions, and relate them to our abstraction bounds derived above. Figure~\ref{applied_abstractions} summarizes the two abstraction approaches.

\begin{figure}[H]
\begin{center}
\includegraphics[scale=.4]{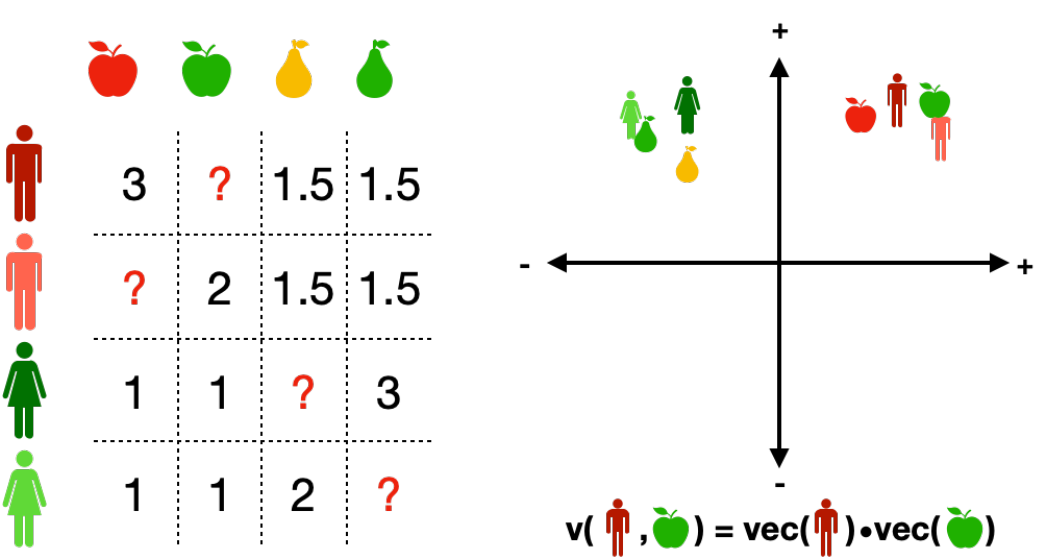}
\hspace{25px}
\includegraphics[scale=.35]{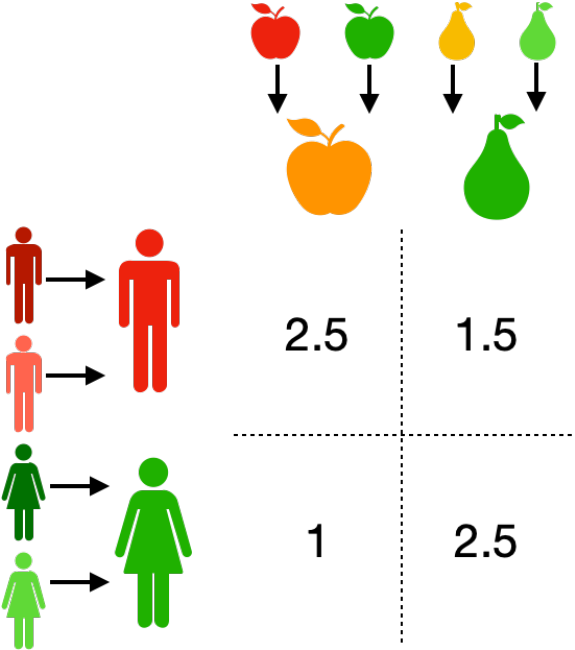}

\end{center}
\caption{Two abstraction methods that can be used together to reduce information and computational constraints. Left: Low-rank matrix completion uses observed data to discover a latent vector for each item and buyer. Unobserved valuations are approximated using the dot product of these vectors. Right: Representative buyer/item abstractions collapse multiple buyers and/or multiple items into single representative agents.}
\label{applied_abstractions}
\end{figure}

Recall that most of the bounds (with exception of the Nash social welfare) related the quality of the equilibrium $(\hat{p}, \hat{X})$ to $\|\Delta V\|_{1,\infty} $. This norm is computationally difficult to work with in practice since it is non-smooth. Instead, we will consider a much better understood matrix norm for practical applications, the Frobenius norm, which is defined as $\| \Delta V \|_{F} =(\sum_{i,j}\Delta v_{i,j}^2)^{1/2}$. 

Using the Frobenius norm implies that we have to settle for a less tight bound given by:
\[\|\Delta V\|_{1,\infty} = \max_{i \in [n]} \|\Delta v_i\|_1 \leq \max_{i \in [n]} \sqrt{m}\|\Delta v_i\|_2 \leq \sqrt{m} \|\Delta V\|_{F}.\]

Using the Frobenius norm gains us the ability to use well-known techniques to construct our abstractions. We will use low-rank matrix completion to learn missing entries and here it is well known that if $\hat{V}$ is the optimal $k$-rank approximation to $V$ then $\| \Delta V \|_{F} = (\sum_{i=k}^{\min(n,m)} \sigma^2_i)^{1/2}$ where $\sigma_i$ is the $i^{th}$ singular value of the original matrix~$V$---that is, the Frobenius norm is related to the `cut off' singular values of the matrix. 

Furthermore, if our matrix completion is being done from samples (as is common in many online applications) we can use a hold-out test set to estimate $\| \Delta V \|_{F}.$

In addition to rank-$k$ approximation, we will also use $k$-means clustering to reduce the size of $V$; this technique can also be viewed as a way of finding an approximation that minimizes the Frobenius norm~\citep{bauckhage2015k}.

\begin{remark}
  While our bound tells us about the quality of the allocations and prices, it tells us little about how close the values themselves are to the exact solution. It is known that the equilibrium $(p,X)$ is continuous in $V$ \citep{megiddo2007continuity}. Thus, we can make the allocation and prices arbitrarily close by making $\|\Delta V\|$ sufficiently small. However, unlike quantities like regret, envy, etc., there is no linear relationship between $\|\Delta V\|$ and how far the abstracted price vector and allocation matrix is from the equilibrium under $V$. This is an interesting area for future work. 
\end{remark}
 
\subsection{Low-Rank Matrix Completion}
We begin with the information problem, that is, the case where not all entries of $V$ are known to us. Without some imputation of the missing entries, we cannot compute an equilibrium. If every $v_{ij}$ is completely independent and unpredictable from other valuations of the same buyer or other valuations of the same item by other buyers, then we have no hope of filling in missing entries in a sensible way. However, in most real world situations this is likely not the case and there is shared information across entries of the valuation matrix. To fill in the missing entries we can use standard techniques from matrix completion (e.g., \citealp{recht2011simpler}).

A standard method for matrix completion is as follows: given a set of observations $\mathcal{O}$ of elements from a matrix $V$ with generic element $v_{ij}$, we try to find a set of vectors $vec(i) \in \mathbb{R}^d$ for every buyer $i$ and $vec(j) \in \mathbb{R}^d$ for every item $j$ to solve $$\min_{vec} \sum_{v_{ij} \in \mathcal{O}} (v_{ij} - vec(i) \cdot vec(j))^2.$$ After we fit this model, we can construct a now complete matrix $\hat{V}$ of the original dimensionality with $\hat{v}_{ij} = vec(i) \cdot vec(j)$ and use this in place of $V$ for our task of interest. The solution ends up effectively minimizing $\min_{\hat{V}} \| V - \hat{V} \|_{F}$  over rank $d$ matrices $\hat{V}$.\footnote{This formulation is non-convex and thus does not have guarantees. However, it is well-known to perform well in practice, and it is much more scalable than the convex relaxation \citep{ge2016matrix,bhojanapalli2016global}} Note that matrix completion is not guaranteed to give positive $\hat{v}_{ij}$ even if all observed $v_{ij}$ are positive, so we replace all $\hat{v}_{ij} \leq .001$ with $.001.$ 

\subsection{Representative Market Equilibrium}

The issue of computation constraints when solving for market equilibrium arises because the EG program grows linearly in the size of the matrix $V$. This can lead to issues both with computation time, and especially memory usage. 
If an abstraction can reduce the number of buyers (and/or items), it can make the computation of equilibrium much more efficient. Given a valuation matrix $V$ we consider abstracting the market as follows.

\subsubsection{Proportional Lift}

\begin{algorithm}
  \caption{Representative Market Equilibrium}\label{alg:repr}
\begin{algorithmic}
  \State \text{Input:} Set $\mathcal{O}$ of known valuations with generic element $o_{ij}$ (not necessarily the full matrix), $\hat n$, $\hat m$.
\Procedure{Construct Representatives}{}
\State If matrix is fully available, use the row/columns as the vectors
\State Else solve for vectors for each buyer $\text{vec}(i)$ and item $\text{vec}(j)$ 
\State Use $k$-means clustering on the buyer vectors with $\hat{n}$ centroids
\State Use $k$-means clustering on the item vectors with $\hat{m}$ centroids
\State Use centroids of $k$-means as vector representations of representatives
\State Set representative market valuations as dot products of above vectors
\EndProcedure

\Procedure{Compute Representative Supplies/Budgets}{}
\State Sum budgets of buyers assigned to representative $i$ by $k$-means to get budget
\State Sum supplies of items assigned to representative $j$ by $k$-means to get supply
\EndProcedure

\State Compute market equilibrium in the representative market
\end{algorithmic}
\end{algorithm}

Algorithm~\ref{alg:repr} constructs a representative market equilibrium (RME) which we denote as $(\tilde p_{rep}, \tilde X_{rep})$. The RME is of a different dimension than the original market. The question is now how to lift the prices and allocations to the original items and buyers. Prices are simple, we can just assign to every item $j$ the price of its representative item from the RME.

Allocations are more difficult. We first consider a proportional lift. We perform the lift in two steps. First we begin with $\tilde X_{rep}$ and construct a new allocation $X'$ of dimension $\hat{n} \times m$ with $x'_i$ being the allocation of \textit{original} items to each representative buyer. Each representative buyer receives a supply weighted share of each item. Let $r(j)$ be the representative item for real item $j$. Then the amount of item $j$ that is allocated to representative buyer $i$ in this step is 
$$x'_{ij} = \dfrac{s_{j}}{\sum_{k \in r(j)} s_k} \tilde x_{i r(j)}.$$

Given $X'$ we finish the lift by now splitting items across individuals proportionally to their budget. We denote the final allocation as $\hat{X}\in \R_+^{n\times m}$. Let $r(i)$ denote the representative buyer for real buyer $i$. The allocation of item $j$ to real buyer $i$ is $$\hat{x}_{ij} = \dfrac{b_i}{\sum_{k \in r(i)} b_k} x'_{r(i)j}.$$

Because everything is allocated proportionally we have that the supply constraint binds and thus the allocation is a full allocation. The proportional lift has the advantage that it is simple to compute and that our bounds can be applied directly. The allocations from the proportional RME are equivalent to those which would result if we constructed an approximate valuation matrix $\hat{V}$ by getting vector representations for each buyer/item, getting their representative abstraction, and replacing $v_{ij}$ by $\text{vec}(r(i)) \cdot \text{vec}(r(j)).$

\subsubsection{Recursive Lift}

Proportional lifting is arguably the most direct and natural way to lift a solution. However, it makes no attempt to utilize the fact that the different buyers mapping to a single representative buyer are generally heterogeneous, even if similar.
Next we present a lifting procedure which exploits the heterogeneity within each group of buyers assigned to the same representative buyer to improve on proportional lifting.

First, let us define $X^{R}$ to be an allocation of original items to representative buyers. We define a lift $L(X^{R})$ to be an allocation to the original buyers such that they are only allocated items from their respective representative buyer. We let $L_i(X^{R})$ be the bundle of goods received by buyer $i$ in the lift. We denote by $r(\tilde i)$ the set of buyers assigned to representative buyer $\tilde i$.

\begin{defn}
  $L(X^R)$ is a \textit{local lift} if for any representative buyer $\tilde i$, we have that $\sum_{i \in r(\tilde i)} L_i(X^R) \leq x^R_{\tilde i}$.
\end{defn}

Secondly, we shall use the less formal restriction that a local lift should be computable using only $X^R$ and the set of valuations $\{v_i : i \in r(\tilde i)\}$. 

To be concrete, consider the following example:
\begin{example}
  There are three buyers and three items, all with supplies and budgets equal to one. A representative market is created with two buyers: $\{(1), (2,3)\}$ and two items: $\{(1), (2,3)\}$. So, the first buyer and item remain unabstracted, but the latter two buyers and items are merged into a single representative buyer and item. 
  Now, say that the solution to the abstraction is to split each of the two abstracted items in half. Then, the proportional lift is for buyers $2$ and $3$ to receive $\frac{1}{4}$ of each of the three original items. A local lift on the abstracted buyer is any reallocation of the items such that buyers $2$ and $3$ receive at most half of each item in aggregate.
  \label{example: local lift}
\end{example}

The concept of a local lift is natural when we want to consider using parallelization for the computation of the lift. Locality means that only information about the allocation a representative buyer received is used to compute the re-allocation to the buyers represented by the representative buyer. Thus, we can solve the main abstraction and then solve each component of the lift (defined by each representative buyer) in parallel without any message passing. The proportional lift is a local lift.

Recall that the original objective function for the EG equilibrium is the maximization of Nash social welfare. Thus, we now consider a lift which is optimal in the family of local lifts in the sense that it is the NSW-maximizing one.

The lift works as follows: For each representative buyer, we look at the allocation of original items to that representative buyer in the abstraction equilibrium, denote this by $\tilde{x}^R_{\tilde i}$. We then construct a new \emph{local} market for each representative buyer $i$. The buyers in the local market are individuals in $r(\tilde i)$, their budgets are their original budgets, and supplies are $\tilde{x}^R_{\tilde i}$. 

The lift $L(X^R)$ is the solution to 
$$\max_{X} \left\{ \sum_{i \in r(\tilde i)} \log (v_i \cdot x_i ) : \sum_{i \in r(\tilde i)}{x_i} \leq \tilde{x}^{R}_{\tilde i}\right\} .$$
This lift yields a market equilibrium in the local market. We call this approach \emph{recursive representative market equilibrium} (RRME).
From the way the lift is constructed it is obvious that a local lift maximizes Nash social welfare among all local lifts if and only if it is an RRME.

In Example~\ref{example: local lift}, the RRME is constructed by creating a new market with buyers $2$ and $3$ in it, and where the supply of each item is $1/2$. Then, we lift their allocation by finding a market equilibrium allocation of this new market. More generally we would need to solve for multiple markets, but here there is only one representative buyer that represents multiple original buyers.

We can also show that RRME is guaranteed to improve on several metrics relative to the proportional lift: Pareto gap, regret, MMS gap, and NSW by Theorem~\ref{thm:nsw}. Intuitively, this is due to the utility guarantee provided by market equilibrium, which says that each buyer does better than their budget-proportional allocation of each item.

\begin{theorem}
  Let $(p,X)$ be a solution obtained from proportional allocation from solving a representative market instance. Let $(p,X^R)$ be the solution obtained by applying RRME to $(p,X)$. The RRME solution leads to weakly lower Pareto gap, regret, and MMS gap, as well as weakly higher NSW.
  \label{thm:rrme}
\end{theorem}
\begin{proof}

First we note the following simple fact which holds for any agent $i$: the utility of $i$ under $X^R$ is weakly greater than that under $X$, i.e., $v_i \cdot x^R_i \geq v_i\cdot x_i$. This is because (a subset of) $X^R$ is a market equilibrium in the recursive market for the corresponding $\tilde i$, and an agent is guaranteed to get at least the value of the budget-proportional allocation in any market equilibrium.

 The Pareto gap is the value of a linear program that maximizes social welfare (minus current welfare) subject to the constraint that each agent is weakly better off. Since utilities are greater in $X^R$ this is a strictly more constrained problem than for $X$, and thus the value, i.e., the Pareto gap, is lower.

 Since we keep prices the same the optimal bundle $x_i^*$ for each agent remains the same for $(p,X)$ and $(p,X^R)$. Thus the only affected part of regret is the negative term, which is weakly greater under $X^R$ since utilities are weakly greater.

 That the MMS gap is smaller and NSW greater follows directly from each agent having weakly-higher utility.
\end{proof}

While the objectives in Theorem~\ref{thm:rrme} are all guaranteed to improve in the RRME over the proportional lift, the example below shows that envy may not. Intuitively, this is because we may have an agent $i$ who does not envy agent $j$ very much in the proportional lift, but in the RRME $j$ improves very much while $i$ improves very little, thus generating an increase in envy. 

\begin{example}
Consider a 5-buyer-4-item instance with valuations $v_1=v_2=[1.5,1.5,0,0], v_3=[0,0,1+\eps,1-\eps], v_4=[0,0,1-\eps,1+\eps], v_5=[1.5,1.5,1+\eps,1-\eps]$, and budgets and supplies equal to $1$. Now consider the abstraction where buyers $1,2,5$ are clustered to a representative buyer with valuation $v_{\tilde 1}=[1.5,1.5,0,0]$ and budget $3$, and buyers $3,4$ are clustered to a representative buyer with valuation $v_{\tilde 2}=[0,0,1.5,1.5]$ and budget $2$. The representative market equilibrium with equal rates is to assign representative buyer $\tilde 1$ all of items $1$ and $2$, each with price $1.5$ and assign representative buyer $\tilde 2$ all of $3$ and $4$ at price $1$ each. This assignment leads to no envy when we perform the budget-proportional allocation. However, if we apply RME to representative buyer $\tilde 2$ then buyers $3$ and $4$ get assigned all of items $3$ and $4$, respectively. Now buyer $5$ envies buyer $3$, as they could get $\eps$ more utility.
\end{example}

We view the potential increase in envy as a necessary price to pay in order to reap the benefits of locally optimizing each of the representative allocations. In order to avoid potential envy increases, we would need to take into account agents $i \notin r(\tilde i)$ when computing the lift, which would often not be feasible computationally (since the whole point of the representative abstraction was to avoid a convex program whose size is on the order of $n \times m$).


\section{Experimental Evaluation}

\subsection{First-Order Methods for Market Equilibrium}

We initially considered solving the EG convex program directly via the CVXPY package~\citep{cvxpy,cvxpy_rewriting}, which is a package for formulating convex programs in a uniform way, while utilizing underlying optimization solvers such as ECOS~\citep{domahidi2013ecos}, CVXOPT~\citep{andersen2013cvxopt}, and SCS~\citep{donoghue2016conic}. 
However, this approach turned out to have numerical problems once the number of agents and items reached around 130 each, regardless of the underlying solver.

To deal with those scalability issues, we solve the EG convex program via first-order methods. We consider the following Lagrangian relaxation:
\begin{align}
  \label{eq:cp_spp}
  &\min_{0 \leq p \leq \|B\|_1/s} \sum_{i \in [n]} \max_{0 \leq x_i \leq s} B_i \log(v_i \cdot x_i) - x_i\cdot p + s \cdot p\,.
\end{align}

This formulation allows us to apply standard algorithms for solving convex-concave saddle-point problems such as the primal-dual algorithm of \citet{chambolle2011first,chambolle2016ergodic} (henceforth referred to as PD). PD is an algorithm for solving problems of the following form (we omit some unnecessary generality):
\begin{align}
  \label{eq:pd_spp}
  \min_{x\in \cal X} \max_{p\in P}\ \mathcal L(x,p) \coloneqq x^TK p + f(x) + s\cdot p
\end{align}
where $K$ is a matrix with bounded norm $L = \|K\| = \max_{\|x\|\leq 1, \|p\|\leq 1}\ x^TKp$, and $f$ is a proper, lower semicontinuous convex function with a Lipschitz-bounded gradient, i.e., $\|\nabla f(x) - \nabla f(x')\| \leq L_f\|x - x'\|$ for all $x,x'\in \cal X$.

An iteration of PD is as follows:
\begin{align*}
  (\hat p, \hat X) =& \mathcal{PD}_{\tau,\sigma}(\bar x, \bar p, \tilde x, \tilde p) \text{ where } \\
  &\begin{cases}
    \hat x = \arg \min_{x\in \cal X} f(\bar x) + \langle \nabla f(\bar x), x-\bar x \rangle + x^T K\tilde p + \frac{1}{\tau}\|x-\bar x\|\,, \\
    \hat p = \arg \min_{p\in P} \frac{1}{\sigma}\|x-\bar x\| -  \tilde x^T Kp\,.
  \end{cases}
\end{align*}

The algorithm repeatedly calls $\mathcal{PD}$ to generate a sequence of iterates as follows (note that the function invocation depends on the computed value, this is implementable because $x^{t+1}$ is computed before $p^{t+1}$):
$$
(x^{t+1}, p^{t+1}) = \mathcal{PD}_{\tau, \sigma} (x^t, p^t, 2x^{t+1} - x^t, p^t).
$$

The average iterates $\bar x=\sum_{t=1}^Tx^t$ and $\bar p=\sum_{t=1}^Tp^t$ converge to a saddle-point solution. The performance of the algorithm is usually measured in terms of reducing the \emph{saddle-point residual}: $\max_{p} \mathcal L(\bar x, p) - \min_x \mathcal L(x, \bar p)$.
A natural question is how quickly this error measure is driven to zero. 

 \citet{nesterov2018computation} study a formulation similar to \eqref{eq:cp_spp}, but where prices are unbounded, and they arrive at an algorithm where the saddle-point residual goes to zero at a rate of $O(1/\sqrt{T})$, but achieve additional auction-like properties in terms of the resulting dynamics. Leveraging the structure of our problem gives us even better worst-case bounds on convergence. Because in practice such worst case bounds are rarely reached and because these theoretical results are not the primary focus of our paper we discuss them informally here for completeness and relegate the formal statements to the Appendix:

\begin{answer}
Under mild assumptions the PD algorithm reduces the saddle-point residual of the EG problem at a rate of $O\left(\frac{n^{7/2}m^2}{T}\right)$.
\end{answer}

The expression above makes some simplifying assumptions for readability. We include the full treatment in the Appendix. From the perspective of abstraction, the most important thing is that we have a $O(n^{7/2}m^2)$ dependence on problem size. For this simplified rate we can give a nice expression for the savings in worst-case convergence rate due to our representative buyer/item abstraction:

\begin{answer}
Let $n$, $m$ be the number of buyers/items in the original market and $\hat{n}$, $\hat{m}$ be the number in the abstracted market. Under mild assumptions the representative buyer/item abstraction with proportional lift has a worst case convergence rate that is $O(\frac{\hat{n}^{7/2} \hat{m}^2}{n^{7/2} m^2})$ of the convergence rate of the full problem. If we use the recursive lift instead the relevant speedup is $O(\frac{\hat{n}^{5/2} \hat{m}^2}{n^{5/2} m^2}).$
\end{answer}

While asymptotic analysis does not tell the whole story, the results do suggest that we should expect larger savings in computational complexity for `long and thin' markets (those with many buyers and relatively few items) than those with many items and few buyers, which is precisely what we see in our experiments. 

In practice, the primary concern is the memory usage of PD, as well as the cost per iteration. For the optimization formulation that uses PD, a variable is needed for every buyer and item pair, and thus memory usage is linear in the market size. Similarly, each step of PD can be computed in linear time for our setting.
In practice, we thus get an asymptotic linear improvement to both memory usage and cost per iteration. 
In theory, we get much more than a linear improvement in total running time, based on the above estimates of the number of iterations needed in the worst case. However, in practice the number of iterations needed does not seem to grow much with the size of the market, and so the primary practical savings are those due to the memory and cost-per-iteration improvements.

In extremely large cases (e.g., online ad systems with millions of buyers and billions of impressions) it would not even be feasible to run PD on the original problem instance (typically due to memory usage), and so abstracting down to a manageable size becomes compulsory.


\subsection{General Analysis Plan}
We now evaluate the abstraction approaches above on several real datasets. We first discuss the analyses we perform on all datasets then we discuss each dataset in detail and give results.

For each dataset we compute the EG equilibrium in the full market and compare to the prices/allocations we generate using abstractions. We set budgets to $1$ and supplies to be such that there is one item per person (note that since we are in the divisible case, this only affects the prices/allocations up to a scaling factor).

We vary three properties of abstractions jointly. First, we replace the valuation matrix $V$ with a rank $k$ approximation for various $k$. To compute these low-rank representations we use the singular value decomposition (SVD). Second, we consider compressing the large number of buyers into a smaller set of representative buyers. We refer to this as abstraction coarseness and measure it as a percentage of the original market size. Thus, a $40\%$ abstraction is one which replaces 7200 original buyers with 2880 representative buyers via the $k$-means procedure above. Third, we compare the use of a proportional split and the recursive splitting. The goods were not abstracted in the first two datasets (but see Section~\ref{sec:movielens} where we do abstract both sides).

In the results reported here we replace abstracted valuations (e.g., low-rank estimated ones) that are below $0$ with $0.01.$ This means that we do not allow our abstractions to turn any player into a `dummy player' that receives no items.

We measure all of our theoretical quantities of interest: regret of each individual given each allocation (we normalize this by the maximum utility of the allocation), an individual's envy (again, normalized by the utility that an individual would receive from the envied bundle), Nash social welfare (normalized by the Nash social welfare of the unabstracted market), Pareto optimality (normalized by the utility in the Pareto-optimal allocation), maximin share (normalized by utility being achieved).

We also look at one quantity for which we do not have bounds: total welfare/efficiency of the allocation (normalized by the total welfare of the unabstracted market). Note that market equilibrium makes no pretense of maximizing efficiency (indeed it only guarantees Pareto optimality and, in the case of EG equilibrium, Nash social welfare). 

\subsection{Dataset: Jester1}
\begin{figure}[t]
\begin{center}
\includegraphics[width=1.50in]{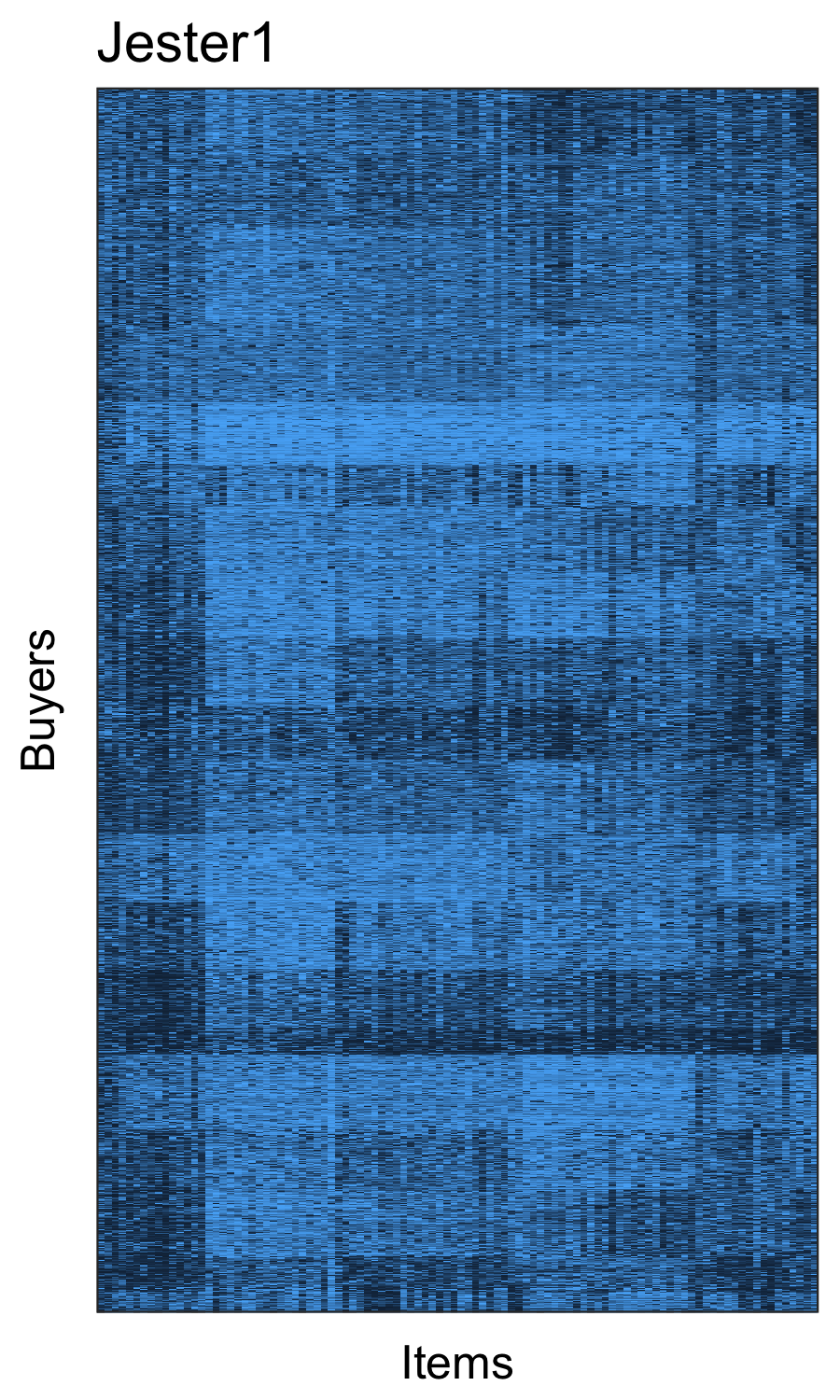}
\includegraphics[scale=.50]{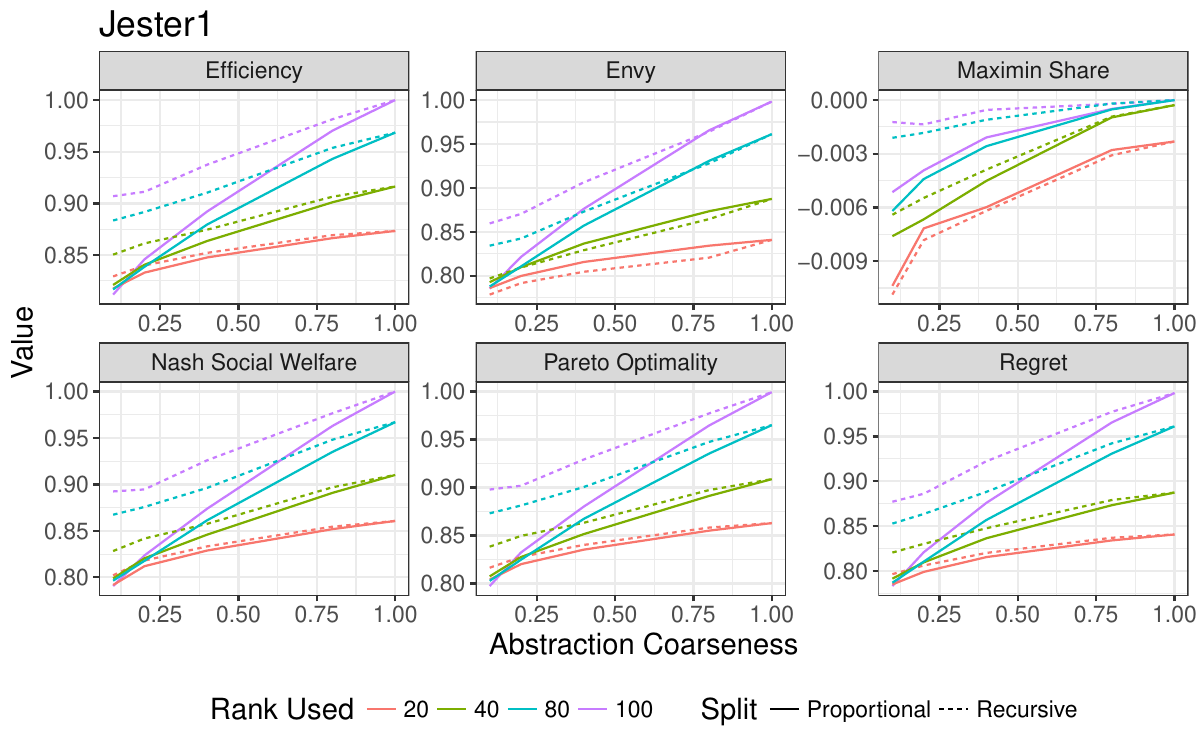}
\end{center}
\caption{Market created by using valuation matrix from the Jester1 dataset with 7200 buyers and 100 items (left panel). There is clear structure suggesting compressibility. In this experiment we apply representative buyer abstraction only. Equilibria computed in quite coarse abstractions maintain good properties (right panel).}
\label{jokes_results}
\end{figure}

We begin by considering an existing dataset used for the evaluation of recommender systems. In typical recommender system datasets individuals give a rating (rather than a monetary evaluation) for each object. Nevertheless, standard recommender system datasets are still useful examples of valuation matrices for two reasons: first, when equilibrium computation is used for market design/fair division (e.g., CEEI and related algorithms at Wharton or Spliddit) the budget given to each individual is only virtual currency and, second, allocations in EG equilibrium are not affected by the scaling of an individual's utilities so they are relatively robust to different ways individuals may interpret ratings.

The first dataset we consider, Jester 1, contains the evaluation of 100 jokes by over 79,000 individuals~\citep{goldberg2001eigentaste}. We extract a submatrix of 7200 individuals that have rated all of the jokes giving us a complete market. Ratings in Jester are on a continuous scale between $-10$ and $10$ and our theory requires positive valuations so we shift the valuations to be (weakly) positive by shifting the whole matrix by $+10$.

Figure~\ref{jokes_results} left panel shows a representation of the Jester1 valuation matrix with lighter colors representing higher valuation. To show the structure more clearly we normalize each individual's valuations to lie in $[0,1]$ for the figure\footnote{Recall that since EG equilibrium with budgets is equivalent to equal utility rates per item being received multiplying an individual's utility by a constant does not affect the equilibrium allocation.} (not the experiments) and perform clustering on the rows and columns to set the order they will be displayed. There is clear block structure suggesting that we can abstract the market using the representative buyer method effectively. 

The right panel shows our main results. Even a coarse abstraction with recursion (720 representative buyers for 7200 original buyers and 20\% rank compression) can yield an allocation that achieves almost $90\%$ of the Nash social welfare and efficiency of the unabstracted allocation and is almost Pareto optimal (there exists an allocation that improves total utility by 10\% without leaving anyone worse off). Individuals display some regret (the abstracted allocation achieves $85\%$ of the utility they could achieve if buyers optimized given the abstraction prices) and some envy. Finally, we see that a weak notion of fair division, maximin share, is almost completely guaranteed under even the coarsest abstraction.   

\subsection{Dataset: Household Items}
\begin{figure}[h!]
\begin{center}
\includegraphics[width=1.5in]{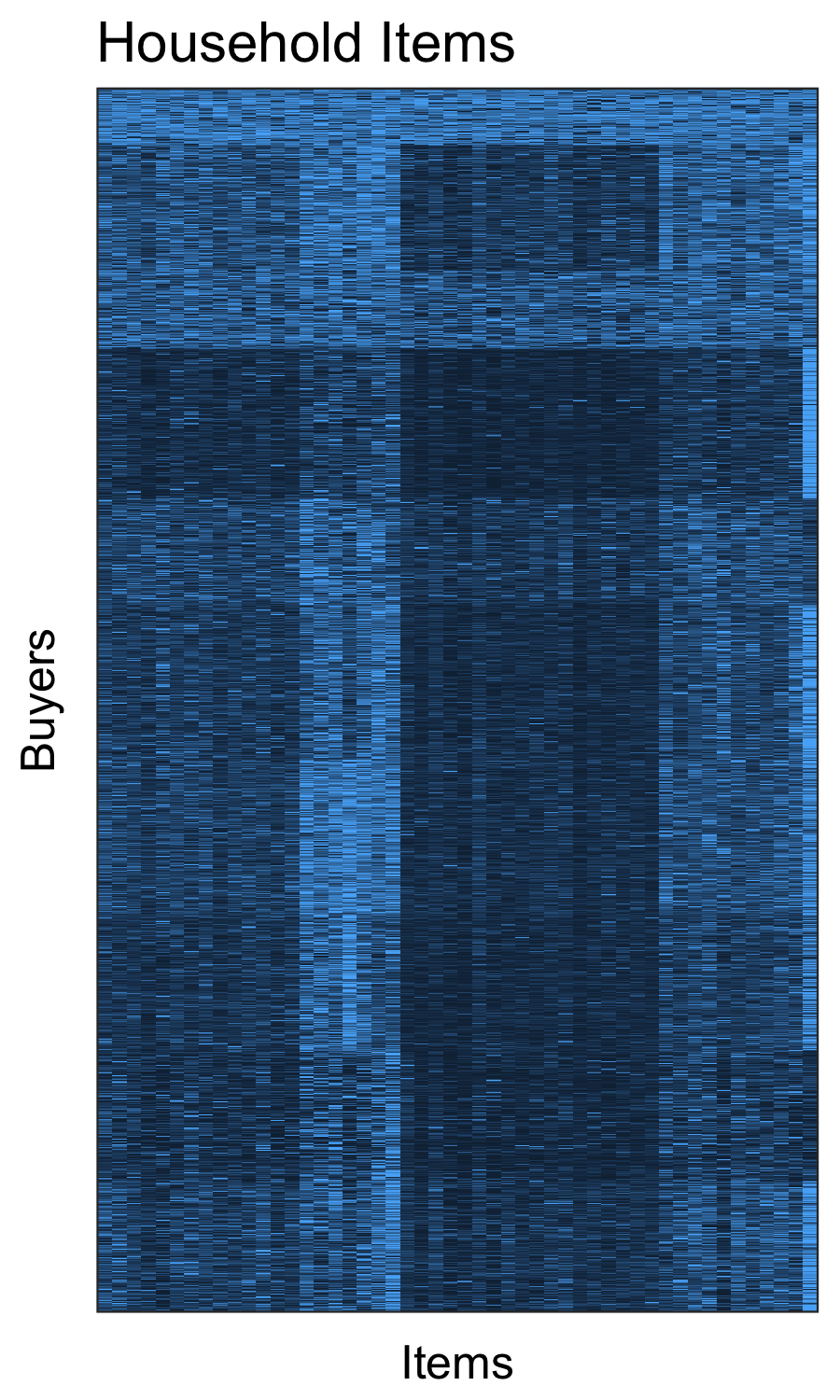}
\includegraphics[scale=.5]{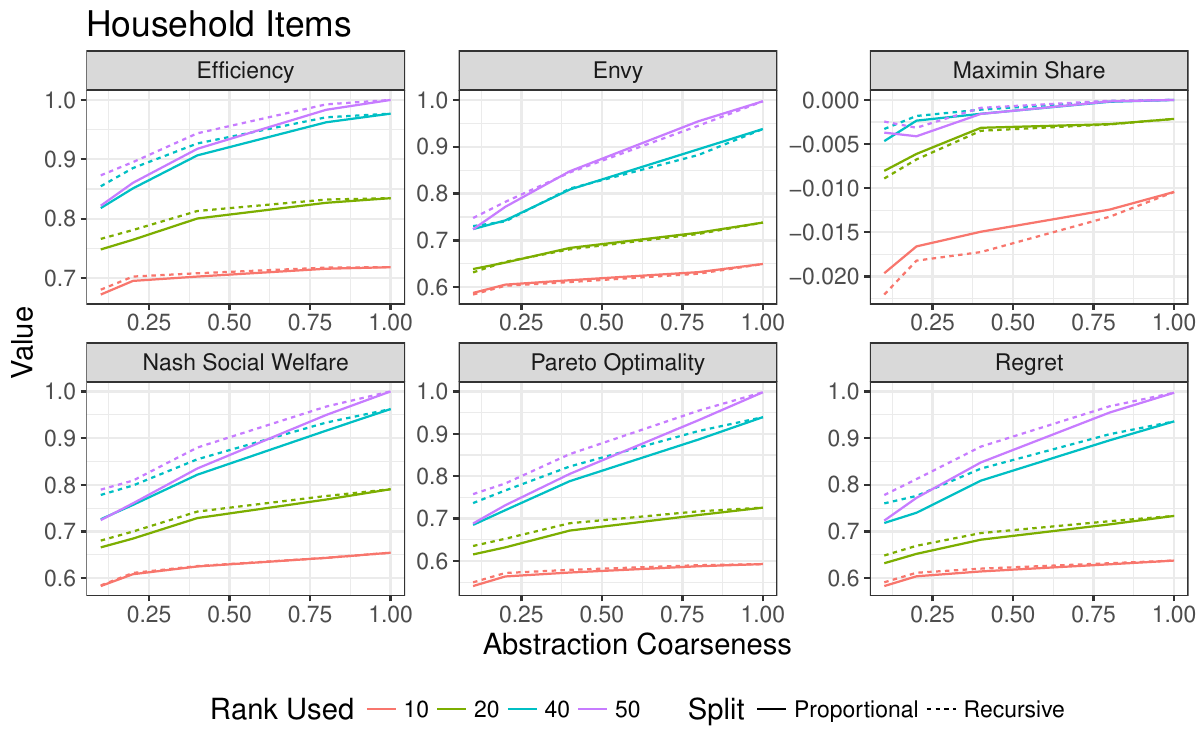}
\end{center}
\caption{Market created by using a valuation matrix from the Household Item dataset with 2876 buyers and 50 items (left panel). There is clear structure suggesting compressibility. In this experiment we apply representative buyer abstraction only. Equilibria computed in quite coarse abstractions maintain good properties (right panel).}
\label{hh_results}
\end{figure}

We also construct a new dataset of individuals who estimated their willingness-to-pay for 50 household items selected from a well-regarded online review site. Unlike in typical recommender datasets where items can often be naturally clustered into `types' (e.g., romance movies, horror movies) we specifically chose items to span a broad range of product categories. (Examples: rain jacket, tool box, toaster, shovel, bluetooth headphones, thermos, blackout shade, bike pump, etc; see Figure~\ref{fig:item_values} for full list of items and average valuations for each one.) Items were chosen to be representative of a large number of categories to make our compression problem more difficult.

In a survey conducted on an online labor market, individuals from the US were presented with a photo and brief description of each item. To deal with quality issues, specific brands and models for each item were selected from an online review site to be the `best in their class' and participants were informed of this. Items were presented in a random order and participants entered their personal US dollar valuation for each item. At the end of the survey, participants were asked how well they felt they understood the questions/task. We use data from the 2876 individuals (out of 3300) that said they felt the task was natural and they could give a good personal valuation for all of the items.
Here we list the items we used as well as their average valuations by individuals. 

\begin{figure}[h!]
\begin{center}
\includegraphics[scale=.5]{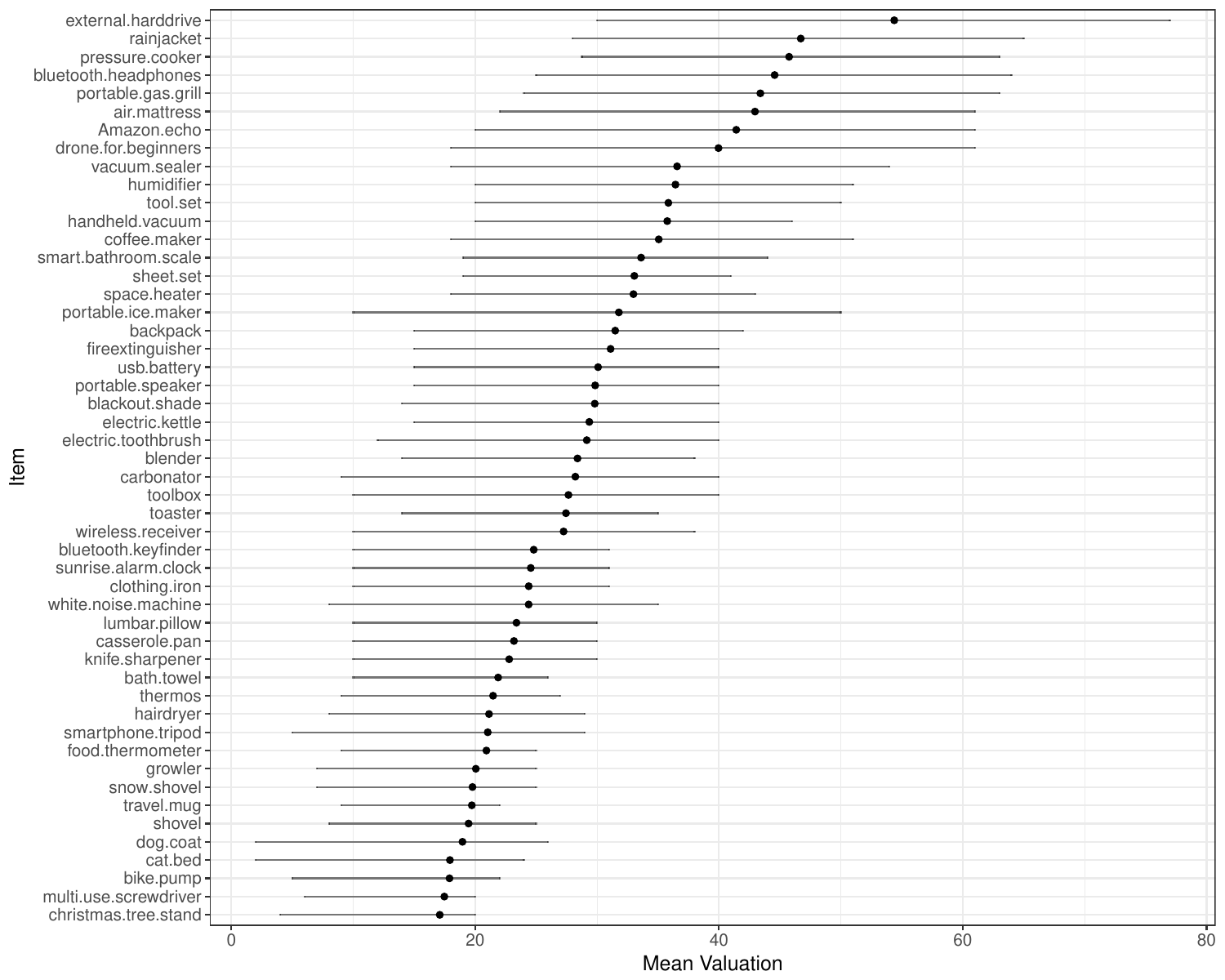}
\end{center}
\caption{Items used in household item survey as well as their average valuations and first and third quantiles.}
\label{fig:item_values}
\end{figure}

For this computational experiment, we apply matrix approximation, representative buyer/item modeling, and recursive reallocation.

Figure~\ref{hh_results} shows the market, represented as in the Jester1 experiment above (left panel) as well as our main results (right panel). We see that even coarse abstractions can yield good properties and that recursive allocation meaningfully improves Pareto optimality and Nash social welfare of the approximations. There is a sharper tradeoff between compression and abstraction performance than in Jester. Recall though that this is because our candidate items were specifically chosen to be of quite different types and so display lower inter/intra item correlations in valuation than typical recommender datasets of, e.g., jokes or movies.

\subsection{Dataset: MovieLens}
\label{sec:movielens}
Both datasets above includes many more buyers than items. We would like to consider a dataset where both sides of the market are large. However, asking individuals to evaluate many hundreds of items would be expensive and likely yield low quality data. Instead, we construct such a market from MovieLens 1M~\citep{harper2016movielens}, a standard dataset for the evaluation of matrix completion algorithms. This dataset contains 6040 individuals with their ratings for a selection of 3952 total movies. In total, the dataset includes 1 million ratings. We treat each 1-5 rating as a valuation.

MovieLens 1M is very sparse (only a small percentage of possible user/movie rating pairs are observed) so we cannot use the data directly to construct a market as we could in the other datasets. 

Instead, we will construct a semi-synthetic dataset by using the data to fit a low-rank approximation to the true rating matrix and doing our analysis on this model. We consider the submatrix which consists of the 1500 movies with the most observed ratings (avg.\ number of ratings: 579) and the top 1500 users who have rated the most movies (avg.\ number of ratings: 425). We use this submatrix as our valuation matrix for the market (Figure~\ref{ml_results} left panel shows the market).

To generate the complete submatrix for the MovieLens 1M market we take the observed ratings denoted as $\mathcal{O}$ with generic element $o_{ij}$. We use PyTorch \citep{paszke2017automatic} and minimize the loss function 
$$\sum_{o_{ij} \in \mathcal{O}} (o_{ij} - \text{vec}(user_i) \cdot \text{vec}(movie_j) + \text{bias}(user_i) + \text{bias}(movie_j))^2$$ 
over $d$ dimensional vectors for each user and movie and $1$ dimensional biases for each user and movie. We randomly split the data into an $80\%$ training set and a $20\%$ validation set and use the validation set to choose $d$ from the set $\lbrace 20, 30, 50, 70, 100 \rbrace$ and the weight decay parameter from the set $\lbrace 10^{-5}, 10^{-4}, 10^{-3}, 10^{-2}, 10^{-1} \rbrace.$  We choose the best performing hyper-parameters from this group via the validation set (d=$20$, weight decay=$10^{-5}$). 

We do not do any extensive tuning nor consider `proper' test set accuracy as we are simply using this model to construct a simple semi-synthetic dataset for our experiment rather than trying to achieve state of the art prediction results.

Because the completion method already uses a low-rank approximation and because we cannot get the ground truth market equilibrium we only test the representative buyer abstraction here. Unlike in the datasets above, we construct both representative buyers and representative items. The right panel of Figure~\ref{ml_results} shows that in this case even an extremely coarse abstraction (150 representative buyers, 150 representative items instead of the original 1500) can achieve outcomes that are quite close those of the full equilibrium. 

\begin{figure}[h!]
\begin{center}
\includegraphics[width=1.5in]{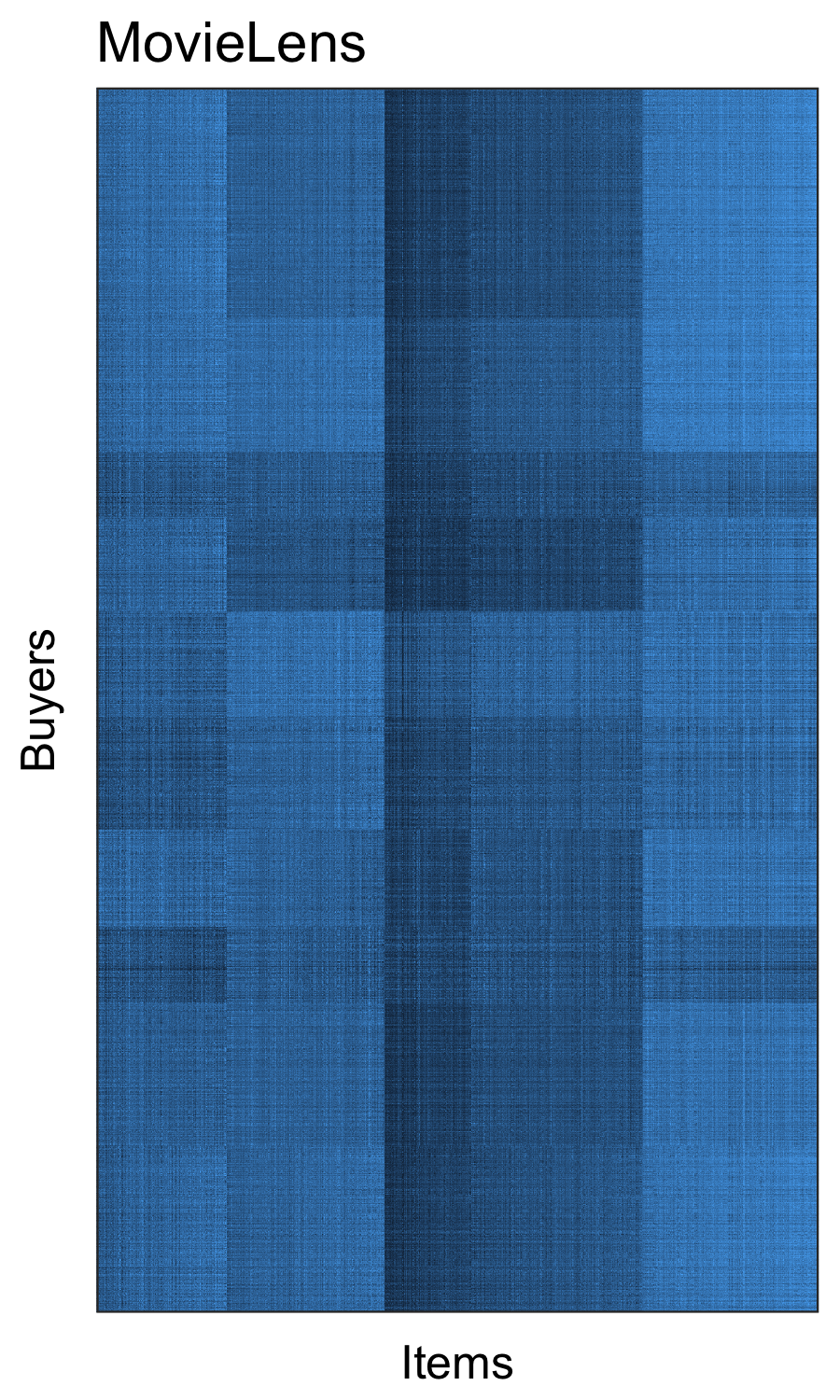}
\includegraphics[scale=.5]{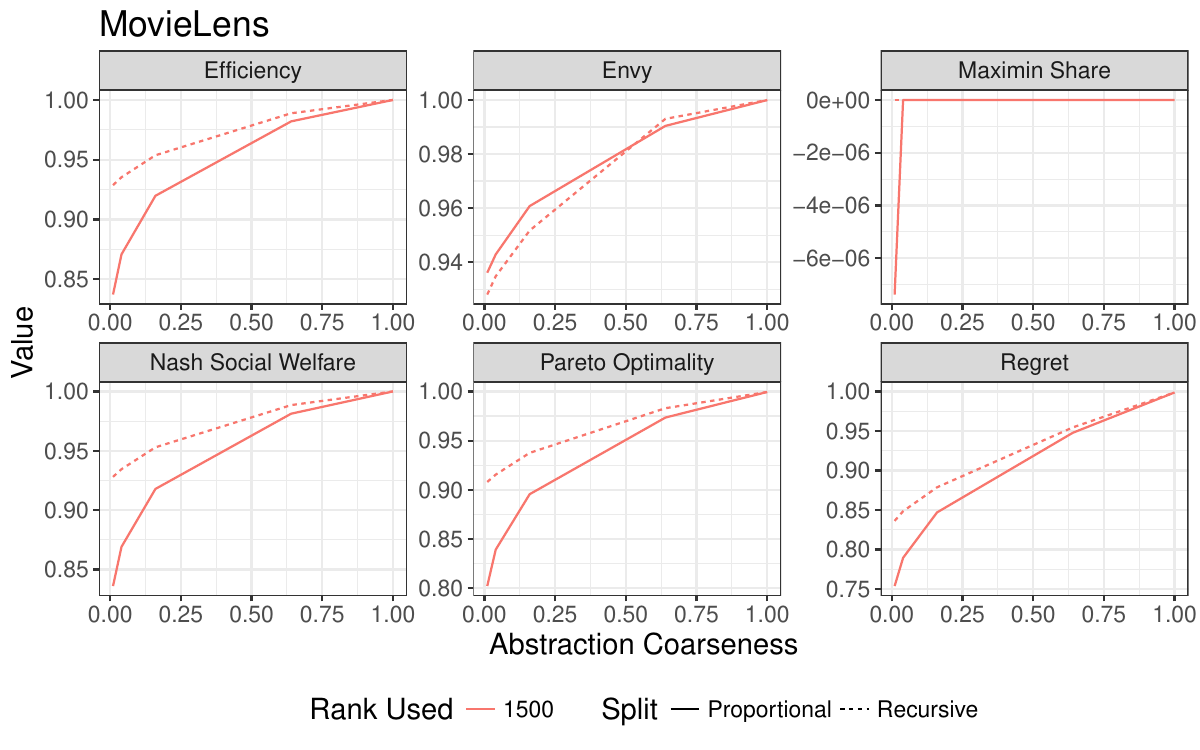}
\end{center}
\caption{Market created by using valuation matrix from the MovieLens dataset with 1500 buyers and 1500 items (left panel). There is clear structure suggesting compressibility. In this experiment we apply both representative buyer and representative item abstraction. Equilibria computed in quite coarse abstractions maintain good properties (right panel).}
\label{ml_results}
\end{figure}

\subsection{A Large Market}

In each of the experiments above, we used markets where we could compute the equilibrium for the full instance. We now consider an experiment for a market that is too large for us to compute the equilibrium for the full instance. In this case we cannot compute metrics like the NSW loss or the Pareto gap but we can still compute individual-level metrics like envy and regret. 

As a large dataset we consider the larger \emph{MovieLens 10 million} dataset. Here $69,897$ users rate $8228$ movies, for a total of $\sim 10$ million ratings. We construct the valuation vectors as in the smaller MovieLens markets above except using a $100$ dimensional embedding.

This problem is too large to solve the full instance on our commodity hardware. Instead, we solve a $4,000 \times 4000$ abstraction---this is a more than $97\%$ reduction in size from the full market. We compute the allocations using both the proportional and recursive lifts. 

Figure~\ref{big_ml_results} shows the results of sampling $1,000$ buyers and computing their envy and regret associated with the lifted allocations and prices. We attain average regret near $16\%$ for the proportional lift and $14\%$ for the recursive lift, and correspondingly small average envy values. The worst values of regret and envy are not much higher.

\begin{figure}[h!]
\begin{center}
\includegraphics[scale=.5]{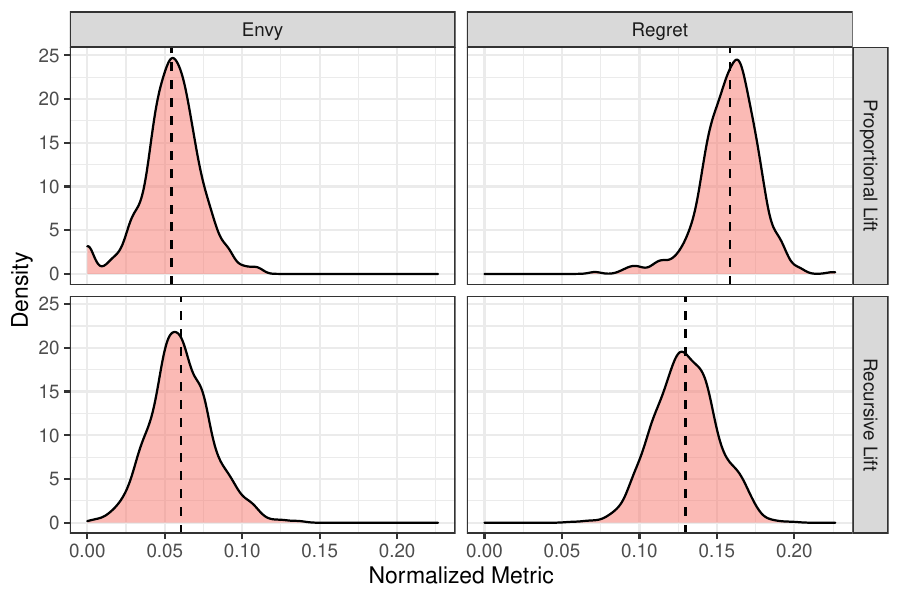}
\end{center}
\caption{In the large MovieLens market, we cannot solve for the full equilibrium, however, we can compute individual levels of regret and envy for some sampled buyers. We see that a $97\%$ reduction in problem size is associated with a relatively small regret and envy in both the proportional and recursive lifts.}
\label{big_ml_results}
\end{figure}

\subsection{Quasi-Linear Markets}

So far we have focused on the case where budgets are useless outside the market. This is the setting of interest for market designers that are using equilibrium as a tool for fair division of items (e.g., CEEI) and also marketing managers in ads marketplaces. However, our approximation results are also applicable to quasi-linear Fisher markets---those where leftover budget has use outside the market. In our household items dataset (unlike in Jester and MovieLens) individual valuations are given in dollars. Although quasi-linear valuations are not the main focus of this work, we present a preliminary experiment to investigate whether abstractions are useful in a quasi-linear market. 

We use the same household items setup as above but now we consider quasi-linear utility and give each individual a budget of $\$ 10$. The quasi-linear case is slightly different from the standard one in several other ways. First, we know that there exists an equilibrium that can be found via a convex program but that this program no longer maximizes Nash social welfare. Second, doing a recursive lift generates multiple prices per item (different prices within each cluster) and is thus not interpretable anymore. Thus, we do not look at either of these in this experiment. 

In addition, because now money is `real' rather than existing solely for the sake of the mechanism, now we also consider an additional metric of interest: how well abstraction prices match up with true equilibrium prices. To measure this we consider the normalized price accuracy given by 
$$1 - \dfrac{\sum_{j} (\hat{p}_j - p_j)^2}{\sum_j p_j^2}.$$ 
A value of $1$ means prices are perfectly predicted. It is known that prices are continuous in the $V$ matrix, so for a given desired level of price accuracy there exists a sufficiently small $\Delta V$ that attains that level. However, like efficiency or weak Pareto optimality, prices are not an LQFS and thus we have no formal guarantees about their accuracy when $\Delta V$ is not small.

Figure~\ref{hh_ql_results} show our main results. As with the standard Fisher case, we see that even coarse abstractions lead to fairly good results.

\begin{figure}[h!]
\begin{center}
\includegraphics[scale=.5]{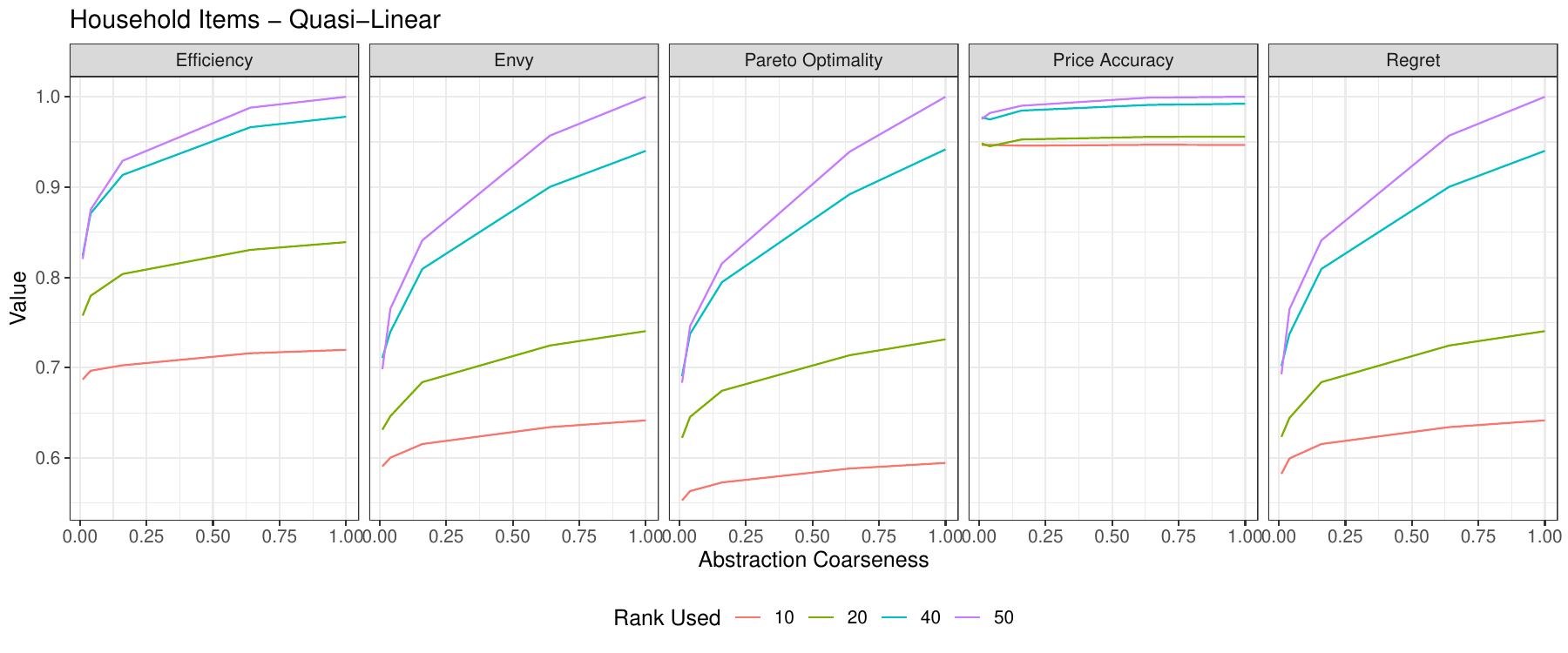}
\end{center}
\caption{Quasi-linear Equilibria computed in quite coarse abstractions maintain good properties.}
\label{hh_ql_results}
\end{figure}


\section{Conclusion and Future Directions}

Computing market equilibria is a difficult problem. We have shown that the method of abstraction---solving a coarser problem and lifting the solution---can be used to reduce the information requirements and computation requirements of equilibrium computation. In addition, we have introduced a new dataset that we hope others can use for work on fair division.

There are many future directions to expand this research. In this work we looked at Fisher markets which assume additive valuations. This rules out situations where goods can be complements or substitutes. In the case of such preferences the information required to compute a market equilibrium becomes extremely large as we need to know the valuations of individuals for every possible combination of goods \citep{porter2003combinatorial}. However, recent work has begun to explore representing preferences with complements and substitutes in low-rank vector format \citep{ruiz2017shopper,peysakhovich2017learning}. An interesting future direction is combining such techniques with abstraction methods. Generalizing beyond linear valuations could lead to greater scalability for problems such as public decision making, which have recently been related to market equilibria~\citep{garg2018markets}.
Another interesting direction would be to apply abstraction techniques to more general market equilibrium problems under utilitarian welfare maximization.

A problem related to our setting is the \emph{cake cutting} problem, where the goal is to allocate (WLOG) the interval $[0,1]$ to agents with heterogeneous preferences over subsets of the interval~\citep{edmonds2011cake,chen2013truth,balkanski2014simultaneous,aziz2016discrete}. This setting is significantly more complicated than linear Fisher markets.
It would be interesting to investigate how our abstraction approach can be generalized to this harder setting, although the relationship to matrix decomposition becomes less pronounced.
Some variants of cake cutting can be related directly to Fisher markets. \citet{aziz2014cake} show that cake cutting with piecewise-constant utilities is equivalent to linear Fisher market equilibrium, and \citet{gao2020infinite} show that general piecewise-linear utilities in cake cutting can be handled via Fisher markets with a continuum of items. The results of \citet{aziz2014cake} imply that our results extend directly to piecewise-constant-utility cake cutting, An interesting research direction is to what extent our results can be used to extended to other cake-cutting settings, for example via the relationship established by \citet{gao2020infinite}.

We considered the use of market equilibria as allocation mechanisms (as in, e.g., the literature on fair division). However, another important use of computation of market equilibria is counterfactual estimation as in structural models in economics \citep{berry1995automobile,chawla2017mechanism}. For example, an online marketplace may want to know how prices (and thus revenues) would change if certain market conditions (e.g., supplies, budgets) were to change. Using the method of abstracting large markets to answer such questions is also an important future direction.

We looked at two methods that could be jointly used to create abstractions: low-rank approximation and representative buyer/item modeling. The representative market abstraction speeds up computation but we did not use low-rank structure for anything but filling in missing data. An interesting algorithmic question is whether the low-rank structure can be leveraged to speed up the gradient calculation steps of our first-order methods, for example, by employing recent techniques for fast nearest neighbor search \citep{JDH17}.

Our work fits into the nascent but growing literature on combining techniques from machine learning/AI with classical results from game theory to solve market design \citep{feng2018deep,golowich2018deep}, game abstraction \citep{moravvcik2017deepstack,brown2018deep}, and agent design \citep{lerer2017maintaining,lerer2018learning} problems that cannot be easily solved in closed form. In this work we leveraged standard linear abstractions (low-rank approximation, k-means clustering). A question for future work is whether more complex, non-linear methods, can be used to construct even better abstractions.


\section{Author Biographies}
\paragraph{Christian Kroer}
Christian Kroer is an Assistant Professor of Industrial Engineering and Operations Research at Columbia University, as well as a member of the Data Science Institute at Columbia. His research interests are at the intersection of operations research, economics, and computation, with a focus on how optimization and AI methods enable large-scale economic solution concepts. He obtained his Ph.D. in computer science from Carnegie Mellon University, and spent a year as a postdoc with the Economics and Computation team at Facebook Research. He was a winner of the 2016-2018 Facebook Fellowship in economics and computation, and a runner-up in the 2017 INFORMS Computing Society Student Paper Competition.

\paragraph{Alexander Peysakhovich}
Alexander Peysakhovich is a senior research scientist at Facebook Artificial Intelligence research. His work includes both basic and applied research at the intersection of behavioral science, economics, game theory, and artificial intelligence. He holds a PhD in Behavioral Economics from Harvard."

\paragraph{Eric Sodomka}
Eric Sodomka is a Research Scientist on the Core Data Science team at Facebook, as well as an advisory board member of Data Science Nigeria. His research interests lie at the intersection of economics and computer science. In particular, he is interested in designing markets that facilitate collaboration across academic disciplines, that connect theoreticians and practitioners, and that improve opportunities for underrepresented groups. He currently resides in Edinburgh, Scotland. Prior to that, he worked in Facebook's Menlo Park office, where he managed the Algorithmic Game Theory and Market Design Research group. He received his Ph.D. in Computer Science from Brown University.

\paragraph{Nicolas Stier-Moses}
Nicolas Stier-Moses is a Director at Facebook Core Data Science. His work leverages innovative research to drive impact to the products, infrastructure and processes at Facebook, the company. The group draws inspiration from a rich and diverse set of disciplines including Operations, Economics, Mechanism Design, Algorithms, Statistics, Machine Learning, Experimentation, and Computational Social Science. Between 2014 and 2017, he supported the Economics, Algorithms and Optimization team, which is one of the areas of focus of Core Data Science. Prior to joining Facebook, Nicolas was an Associate Professor at the Decision, Risk and Operations Division of Columbia Business School and at the Business School of Universidad Torcuato Di Tella. He received a Ph.D. degree from the Operations Research Center at the Massachusetts Institute of Technology.

\bibliographystyle{informs2014}
\bibliography{refs}
\newpage
\appendix
\section{Appendix}

\subsection{Linear properties}
\label{sec:lin props}

\citet{balkanski2014simultaneous} defined \emph{linear properties} as properties of allocations that can be represented as an index set $G$ of finitely-many linear constraints:
\[
  \sum_{i,j} A^k_{ij} u_i(x_j) \geq c^k \quad \forall k \in G
\]

such that $\sum_{ij}A^k_{ij} \leq 1$ for all $k\in G$. It is easy to show that linear properties are approximately preserved under $\Vhat$. Here $G$ is allowed to be an arbitrary finite-sized index set of constraints (so for example, $G$ could be the set of inequalities required for no envy).

Consider any $k\in G$, we have:

\begin{align*}
  c^k &\leq \sum_{i,j} A^k_{ij} \hat u_i(x_j)
   = \sum_{i,j} A^k_{ij} (u_i(x_j) - \Delta v_i \cdot x_j)
   = \sum_{i,j} A^k_{ij} u_i(x_j) - \sum_{i,j} A^k_{ij} \Delta v_i \cdot x_j
   \leq \max_{i,j} \hat u_i(x_j)  
\end{align*}
Now we can use the sum-to-less-than-one property of $A^k$ to bound the error:
\begin{align*}
   \sum_{i,j} A^k_{ij} \Delta v_i \cdot x_j
   \leq \max_{i,j} \Delta v_i \cdot x_j
   \leq \max_{i} \| \Delta v_i \|_1
   = \| \Delta V \|_{1,\infty}
\end{align*}

Thus we can use this general bound on linear properties to show that envy and equitability, as well as any other linear property, is preserved up to an additive error $\| \Delta V \|_{1,\infty}$. \citet{balkanski2014simultaneous} also mention proportionality as a linear property, however here one has to be slightly careful in our setting. Because we have two different market instances $V, V'$, it is not generally the case that there exists a \emph{single} set of linear-property constants $A,c$ representing proportionality as a linear constraint, since the $c$ is in general different for $V$ and $V'$. Thus our result on MMS and proportionality being preserved does not follow directly from bounding linear properties. One has to further relate the RHS constants $c$ and $c'$ under the two instances.

\subsection{Proof of Theorem~\ref{thm:lqfs}}

\begin{proof}
  We first show the result for a per-buyer property.
  Let $X'\in \mathcal X^k$, $i$ be any buyer for which the quantification holders under $\hat V$, $i'$ be arbitrary,
  and $i''$ the minimizer of $u_i(x_{i''}')$, we then have
  \begin{multline*}
    \hat u_i(\hat x_{i}) \geq \lambda_1^k  \hat u_i(\hat x_{i'}) +  \lambda_2^k \min_{i'} \hat u_i(\hat x_{i'}') 
    \Leftrightarrow\\
    \sum_j \hat v_{ij}(\hat x_{ij} - \lambda_1^k \hat x_{i'j} -  \lambda_2^k  x_{i''j}') + QL (B_i - \hat p_j(\hat x_i - \lambda_1^k \hat x_{i'} - \lambda_2^kx_{i''}'))  \geq 0,
  \end{multline*}
  where $QL \in \{ 0,1 \}$ is a constant which denotes whether $u_i$ is quasilinear. This is equivalent to
  \begin{multline*}
    \sum_j v_{ij}(\hat x_{ij} - \lambda_1^k \hat x_{i'j} -  \lambda_2^k  x_{i''j}') + QL (B_i - \hat p_j(\hat x_i - \lambda_1^k \hat x_{i'} - \lambda_2^kx_{i''}'))  
    \geq 
    -\sum_j \Delta v_{ij}(\hat x_{ij} - \lambda_1^k \hat x_{i'j} -  \lambda_2^k  x_{i''j}').
  \end{multline*}
  We can rewrite this in terms of $u_i$ to get
  \begin{align*}
    u_i(\hat x_i) - (\lambda_1^k  u_i(\hat x_{i'}) +  \lambda_2^k \min_{i'} u_i(\hat x_{i'}'))  &\geq -\sum_j \Delta v_{ij}(\hat x_{ij} - \lambda_1^k \hat x_{i'j} -  \lambda_2^k  x_{i'j}'), \\
      &\geq -\sum_j |\Delta v_{ij}|\max(\hat x_{ij}, \lambda_1^k \hat x_{i'j} +  \lambda_2^k  x_{i'j}'), \\
      &\geq - \|\Delta V\|_{1,\infty}\max(1, \lambda_1^k +  \lambda_2^k  ).
  \end{align*}
  Since $i'$ was arbitrary this implies
  \begin{align*}
    u_i(\hat x_i) - (\lambda_1^k \max_{i'}  u_i(\hat x_{i'}) +  \lambda_2^k \min_{i'} u_i(\hat x_{i'}'))  \geq - \|\Delta V\|_{1,\infty}\max(1, \lambda_1^k +  \lambda_2^k  ).
  \end{align*}

  Thus the property holds approximately for buyer $i$ under $V$, as the theorem requires.

  Now consider a global property $k$. By definition we have that for any alternative allocation $X' \in \mathcal X^k$
  \begin{align*}
    0 &\leq \sum_{i\in [n]}\beta_i \hat v_i\cdot \xhat_i - \sum_{i\in [n]}\beta_i \hat v_i\cdot \xhat'_i \\
    & = \sum_{i\in [n]}\sum_{j \in [m]}\beta_i \hat v_i\cdot (\xhat_i - \xhat'_i) \\
    & = \sum_{i\in [n]}\sum_{j \in [m]}\beta_i (v_i - \Delta v_i)\cdot (\xhat_i - \xhat'_i) \\
    & \leq \sum_{i\in [n]}\sum_{j \in [m]}\beta_i v_i \cdot (\xhat_i - \xhat'_i)
      + \sum_{i\in [n]}\sum_{j \in [m]}\beta_i \|\Delta v_i\|_1.
  \end{align*}
\end{proof}

\subsection{Proofs of Other Theorems}

\begin{proof}[Proof of Theorem~\ref{thm:nsw}: Bounded NSW]
  We have
  \[
    \text{NSW}(X^*) = \prod_{i \in [n]} v_i\cdot x^*_i
    = \prod_{i \in [n]} (\hat v_i\cdot x^*_i + \Delta v_i \cdot x^*_i)
    = \prod_{i \in [n]} \hat v_i\cdot x^*_i\bigg(1 + \frac{\Delta v_i \cdot x^*_i}{\hat v_i\cdot x^*_i}\bigg).
  \]
  Now we can use the fact that $X^*$ is feasible under $\Vhat$ to note that its value must be less than that of $\Xhat$:
  \[
    \text{NSW}(X^*) 
    \leq \prod_{i \in [n]} \hat v_i\cdot \xhat_i\bigg(1 + \frac{\Delta v_i \cdot x^*_i}{\hat v_i\cdot x^*_i}\bigg)
    \leq \prod_{i \in [n]} \hat v_i\cdot \xhat_i\bigg(1 + \frac{\|\Delta v_i\|_1}{\hat v_i\cdot x^*_i}\bigg)
    = \text{NSW}(\Xhat) \prod_{i \in [n]} \bigg(1 + \frac{\|\Delta v_i\|_1}{\hat v_i\cdot x^*_i}\bigg).
  \]
\end{proof}

\begin{proof}[Proof of Theorem~\ref{thm:rrme}: RRME weakly improves buyer utilities]

  First we note the following simple fact which holds for any buyer $i$: the utility of $i$ under $X^r$ is weakly greater than that under $X$, i.e., $v_i\cdot x^r_i \geq v_i\cdot x_i $. This is because (a subset of) $X^r$ is a market equilibrium in the recursive market for the corresponding $\tilde i$, and a buyer is guaranteed to get at least the value of the budget-proportional allocation in any market equilibrium.

  The Pareto gap is the value of a linear program that maximizes social welfare (minus current welfare) subject to the constraint that each buyer is weakly better off. Since utilities are greater in $X^r$ this is a strictly more constrained problem than for $X$, and thus the value, i.e., the Pareto gap, is lower.

  Since we keep prices the same the optimal bundle $x_i^*$ for each buyer remains the same for $(p,X)$ and $(p,X^r)$. Thus the only affected part of regret is the negative term, which is weakly greater under $X^r$ since utilities are weakly greater.

  That the MMS gap is smaller and NSW greater follows directly from each buyer having weakly-higher utility.
\end{proof}


\subsection{Convergence Rates of RME Computation}

The average iterates $\bar x=\sum_{t=1}^Tx^t$ and $\bar p=\sum_{t=1}^Tp^t$ converge to a saddle-point solution at a rate of $O(\frac{(L+L_f)D_x + LD_p}{T})$, where $D_x,D_p$ is the maximum value of the $\ell_2$ norm over $\cal X$ and $P$ respectively. Note that we upper bound the price vector $p$ by the sum of budgets $\|B\|_1$, and each allocation vector $x_i$ by the supply of each item. This does not change the set of equilibria, as these conditions are all guaranteed to be satisfied in equilibrium.

We now show how to instantiate our SPP~\eqref{eq:cp_spp} in terms of the generic SPP~\eqref{eq:pd_spp}. We have that $\cal X$ is the product of allocation vectors $\times_{i\in [n]}\{x_i : 0 \leq x_i \leq s\}$ over the buyers, and $P=\{p : 0 \leq p \leq \|B\|_1/s\}$ is the set of price vectors. $K$ is an $nm\times m$ matrix representing $\sum_{i\in [n]}p\cdot x_i$, i.e., with a $1$ in each row/column-pair $r,j$ when the row $r$ corresponds to a variable that denotes assigning item $j$ to some bidder. The norm $L$ of $K$ is $\sqrt{n}$, which is achieved by any pair $x,p$ such that $p_j=1$ for some $j$, and $x_{ij}=\frac{1}{\sqrt{n}}$, with all other entries $0$, or by setting $p=\frac{1}{\sqrt{m}},x=\frac{1}{\sqrt{mn}}$.

The logarithm in the objective function presents a challenge because the gradient is unbounded near zero; this problem can be addressed by noting that agents are always guaranteed to receive their MMS value in equilibrium, and thus we can add the additional constraint $v_i\cdot x_i \geq \textrm{MMS}_i$ to the feasible set of each agent, thereby bounding the gradient difference $L_f$ by $\max_{i\in [n], j \in [m]} \frac{v_{ij}B_i}{\textrm{MMS}_i}$. In practice we found that utilities did not approach zero and thus this projection was unnecessary.

The value of $D_x$ is $n\|s\|_2^2$, the maximum is achieved by setting $x_{ij}=s_j$ for all $i,j$.
The value of $D_p$ is $\sum_{j\in [m]}(\|B\|_1/s_j)^2$, the maximum is achieved by setting each price at its upper bound.

Putting together this construction gives an algorithm that converges to a saddle point of \eqref{eq:cp_spp} at a rate of
$$
\max_{p \in P} \mathcal L (\bar x, p) - \min_{x\in \cal X}\mathcal L (x, \bar p)
\leq
O\bigg(\frac{(\sqrt{n} + \max_{i,j} \frac{v_{ij}B_i}{\textrm{MMS}_i})n\|s\|_2^2 + \sqrt{n} \sum_{j\in [m]}(\|B\|_1/s_j)^2}{T}\bigg).
$$

Now we see that if we solve a clustering $(\{C_{\tilde i}\}, \{C_{\tilde j}\})$ we get the following convergence rate:
\begin{align}
O\bigg(\frac{(\sqrt{\hat n} + \max_{i\in [\hat n], j \in [\hat m]} \frac{\hat v_{ij}\hat B_i}{\textrm{MMS}_i})\hat n\|\hat s\|_2^2 + \sqrt{\hat n} \sum_{j\in [\hat m]}(\|\hat B\|_1/\hat s_j)^2}{T}\bigg).
  \label{eq:clustered_convergence_rate}
\end{align}

Furthermore, each iteration is of order $O(\hat n \hat m)$ rather than $O(nm)$. To get a sense of the savings in running time, say that supply and budgets are both $1$. In that case the original problem has runtime cost $O(n^{7/2}m^2)$ (since the third term usually dominates), and thus compressing to $10\%$ problem size leads to a runtime decrease of factor $100$ to $3162$ depending on whether the decrease in instance size is primarily due to fewer items or buyers.

If we apply recursive lift then we have to solve $k$ market equilibrium problems corresponding to each cluster of buyers. The cost of computing the recursive allocations is the sum of the costs of computing each recursive market equilibrium. This can be expressed as a sum over terms similar to \eqref{eq:clustered_convergence_rate}, but where $\hat n, \hat m$ represent the size of the given recursive instance.
Again we assume that budgets and supplies are $1$ to get a sense of runtime savings. We can then use the fact that the union of the buyer clusters is $[n]$ to bound the runtime as $O(n\hat n^{5/2}\hat m^2)$, where $\hat n, \hat m$ corresponds to the size of the recursive market equilibrium which maximizes $\hat n^{5/2}\hat m^2$. Thus even with the additional cost of computing the RME allocation the runtime cost savings are on the order of $100$ to $316$ for the case where each recursive market is a tenth of the size of the original market.


\subsection{Quasilinear variant}

In this section we describe how our results can be instantiated for quasilinear (QL) utilities.

Each buyer $i$ has a utility that now depends on how much leftover money they have. We let $\delta_i$ denote the amount of leftover money.  The utility of buyer $i$ is now $$u_i^{QL}(x_i, \delta_i) = v_i \cdot x_i + \delta_i$$. 

Given prices $p \in \R_+^m$ for the goods, a \emph{QL demand} for buyer $i$ is $$d_i^{QL} (p) = \lbrace \arg \max_{x, \delta : x\cdot p + \delta \leq B_i} v_i \cdot x + \delta \rbrace.$$ The QL demand can be set valued but the maximum reachable utility given a price vector is unique.

Similar to the non-QL case, a market equilibrium can be computed via convex programming (EG-QL):
\begin{equation*}
\begin{aligned}
& \underset{x}{\text{max}}
& & \sum_{i} B_i \text{log} (v_i \cdot x_i + \delta_i) - \delta_i \\
& \text{subject to}
& & \sum_{i} x_{ij} \leq s_{j} \text{ } \forall j
\end{aligned}
\end{equation*}
Note that this convex program is almost identical to the standard EG program; the only difference is the introduction of $\delta_i$ for each buyer $i$ (this convex program for QL utilities was first given by \citet{chen2007note} and later by \citet{cole2017convex}).

Now consider the utility function $\hat u_i^{QL}(x_i, \delta_i) = \hat v_i \cdot x_i + \delta_i$. Since the $\delta_i$ term does not depend on $\hat V$, we have
\[u_i^{QL}(x_i, \delta_i) = v_i \cdot x_i + \delta_i = (\hat v_i + \Delta v_i) \cdot x_i + \delta_i = \hat u_i^{QL}(x_i, \delta_i) + \Delta v_i \cdot x_i,\]
and thus the abstracted utility function behaves exactly the same as in the non-QL case with respect to the abstraction error $\Delta v_i$. We then have that

\begin{theorem}
  Envy, regret, MMS, proportional share, and strong Pareto improvement are LQFS properties under QL utilities, and are thus bounded by $\|\Delta V\|_{1,\infty}$. Weighted social welfare is LQFS and is thus bounded by $\|\hat w \|_1\|\Delta V\|_{1,\infty}$ for weights $\hat w$.
\end{theorem}
The proof is identical to the non-QL case after noting that $\Delta v_i$ behaves exactly the same in both cases. 


\end{document}